\documentclass[letterpaper]{article} 
\usepackage{aaai23}  
\usepackage{times}  
\usepackage{helvet}  
\usepackage{courier}  
\usepackage[hyphens]{url}  
\usepackage{graphicx} 
\usepackage{natbib}  
\usepackage{caption} 
\usepackage{algorithm}

\usepackage{newfloat}
\usepackage{listings}
\DeclareCaptionStyle{ruled}{labelfont=normalfont,labelsep=colon,strut=off} 
\floatstyle{ruled}
\newfloat{listing}{tb}{lst}{}
\floatname{listing}{Listing}
\usepackage{amsmath}
\usepackage{booktabs}
\usepackage{algorithmicx}
\usepackage[noend]{algpseudocode}
\usepackage{color,xcolor} 
\usepackage{amssymb}
\usepackage{autobreak}
\usepackage{caption}
\usepackage{subfigure}
\usepackage{lipsum}
\usepackage{multirow}
\usepackage{bbding}
\usepackage{makecell}
\usepackage{ntheorem}
\usepackage{array}
\newtheorem{theorem}{Theorem}
\newtheorem{lemma}{Lemma}

\newtheorem{Defs}{Definition}
\newtheorem*{proof}{Proof}
\newtheorem{proposition}{Proposition}

\title{Echo of Neighbors: Privacy Amplification for Personalized Private \\Federated Learning with Shuffle Model}

\author {
    Yixuan Liu\textsuperscript{\rm 1,2,3},
    Suyun Zhao\textsuperscript{\rm 1,2,3},
    Li Xiong\textsuperscript{\rm 4},
    Yuhan Liu\textsuperscript{\rm 1,2,3},
    Hong Chen\textsuperscript{\rm 1,2,3}\thanks{Corresponding author: Hong Chen, chong@ruc.edu.cn}
}
\affiliations {
    \textsuperscript{\rm 1} Key Laboratory of Data Engineering and Knowledge Engineering of Ministry of Education, Renmin University of China\\
    \textsuperscript{\rm 2} Engineering Research Center of Ministry of Education on Database and BI\\
    \textsuperscript{\rm 3} Information School, Renmin University of China\\
    \textsuperscript{\rm 4} Department of Computer Science, Emory University\\
    \{liuyixuan, zhaosuyun\}@ruc.edu.cn, lxiong@emory.edu, \{liuyh2019, chong\}@ruc.edu.cn
}

\usepackage{bibentry}

\begin{document}

\maketitle

\begin{abstract}
Federated Learning, as a popular paradigm for collaborative training, is vulnerable against privacy attacks. Different privacy levels regarding users' attitudes need to be satisfied locally, while a strict privacy guarantee for the global model is also required centrally.
Personalized Local Differential Privacy (PLDP) is suitable for preserving users' varying local privacy, yet only provides a central privacy guarantee equivalent to the worst-case local privacy level.
Thus, achieving strong central privacy as well as personalized local privacy with a utility-promising model is a challenging problem.
In this work, a general framework (APES) is built up to strengthen model privacy under personalized local privacy by leveraging the privacy amplification effect of the shuffle model. 
To tighten the privacy bound, we quantify the heterogeneous contributions to the central privacy user by user. The contributions are characterized by the ability of generating “echos” from the perturbation of each user,  which is carefully measured by proposed methods Neighbor Divergence and Clip-Laplace Mechanism.
Furthermore, we propose a refined framework (S-APES) with the post-sparsification technique to reduce privacy loss in high-dimension scenarios.
To the best of our knowledge, the impact of shuffling on personalized local privacy is considered for the first time. We provide a strong privacy amplification effect, and the bound is tighter than the baseline result based on existing methods for uniform local privacy.
Experiments demonstrate that our frameworks ensure comparable or higher accuracy for the global model.
\end{abstract}

\section{Introduction}
Federated Learning (FL) \cite{mcmahan2017communication} is an emerging machine learning paradigm 
that allows multiple clients to train a global model collaboratively while keeping the private raw data of each client locally. 
\textcolor{black}{While not directly sharing private data, recent works indicate that FL by itself is insufficient to preserve privacy of users' data.}
By observing the global model or intermediate parameters during the training process, adversaries can infer the membership of users or even reconstruct training records \cite{fredrikson2015model,zhu2019deep, nasr2019comprehensive, xiong2021privacy}.  These attacks can lead to severe data leakage, hence it is necessary to provide additional protection with strict privacy guarantees for both the global model and local parameters. 
Moreover, in practice, different local privacy levels may be desired depending on users’ privacy preferences. 
A one-size-fits-all approach would either downgrade the model utility or sacrifice privacy protection for certain users. 
Thus, an open problem in FL is how to provide strong central privacy as well as personalized local privacy while maintaining model utility.

\begin{table}
  \centering
  \begin{footnotesize}
  \renewcommand{\arraystretch}{0.95}
  \begin{tabular}{llll} 
    \toprule
   \multirow{2}*{{Methods}} & \multirow{2}*{Personalization} & \multicolumn{2}{l}{\qquad FL Process}            \\
   \cline{3-4}
   \specialrule{0em}{1pt}{1pt}
       &  & Local & \quad Central\\
    \midrule
    PLDP    & \qquad \Checkmark    &   \quad \Checkmark &   \quad Weak   \\
    Uni-Shuffle & \qquad \XSolidBrush  &   \quad \Checkmark &   \qquad \Checkmark   \\
    APES    &   \qquad \Checkmark    &   \quad \Checkmark &   \qquad \Checkmark     \\
    S-APES    &   \qquad \Checkmark    &   \quad \Checkmark &   \quad Strong     \\
    \bottomrule
  \end{tabular}
  \end{footnotesize}
  \caption{Comparison of related work. \Checkmark denotes protected, \XSolidBrush denotes unprotected.}
  \label{table:motivation}
\end{table} 

Several recent works have attempted to address this problem. 
Personalized Local Differential Privacy (PLDP) protects both local gradients and the global model by perturbing gradients with heterogeneous parameters \cite{chen2016private,li2020federated}. The central privacy of the global model is equivalent to the weakest local privacy.
 For achieving both  strong central and local privacy, a potential solution is the shuffle model \cite{bittau2017prochlo}. It amplifies central privacy by permuting data points randomly after local perturbation. However, existing studies on shuffle model only focus on the scenarios where local privacy requirements are assumed uniform (Uni-Shuffle for short) \cite{erlingsson2019amplification, balle2019privacy, girgis2021renyi, feldman2022hiding}. 
 To the best of our knowledge, there is no work that provides both strong central privacy for the global model and personalized local privacy guarantees, while achieving strong utility of global model (cf. Tab.1).

To narrow this gap, we propose \textbf{APES}, a privacy \textbf{\underline{A}}mplification framework for \textbf{\underline{PE}}rsonalized private federated learning with \textbf{\underline{S}}huffle model (cf. Fig. \ref{fig:PerS_framework}). 
APES gains a strong privacy amplification effect. Unlike previous works that just permute data, both data points and privacy parameters are randomly shuffled in APES. Clip-Laplace Mechanism is also introduced to implement the framework without damaging model utility.
 In order to mitigate the privacy-loss explosion problem caused by high dimensions, we propose \textbf{S-APES} which improves \textbf{\underline{APES}} with the post-\textbf{\underline{S}}parsification. The basic idea is to select only informative dimensions of gradients after perturbation and pad the rest, which saves privacy cost. 

To bound the privacy of APES and S-APES, we carefully quantify the obfuscation effects contributed by users with heterogeneous privacy parameters. First, inspired by  \citeauthor{feldman2022hiding}, the central privacy of a specific user is boosted by the rest of the users who generate “echos” of her with heterogeneous probabilities; next, to measure the probabilities, we propose Neighbor Divergence and Clip-Laplace Mechanism for limited  output range and bounded divergence among distinct output distributions by users’ local randomizers; then ``echos'' are transformed into certain form, and a tight privacy bound is derived.


Our main contributions are summarized as follows:

(\romannumeral1) We propose privacy amplification frameworks via shuffle model for personalized private federated learning. APES strikes a better balance between central privacy and model utility with Neighbor Divergence and Clip-Laplace Mechanism. Based on it, improved S-APES enhances the privacy for the high-dimension scene.

(\romannumeral2) We provide theoretical analysis for both privacy and convergence bound of the proposed frameworks. To the best of our knowledge, the shuffling effect on personalized local differential privacy is considered for the first time and a strong privacy amplification effect is yielded. The central privacy bound is tighter than the bound derived by na\"ively adopting existing methods for unified privacy.

(\romannumeral3) Comprehensive experiments are conducted to confirm that APES and S-APES achieve comparable or higher accuracy for the global model with stronger central privacy compared to the state-of-the-art methods without downgrading personalized local privacy guarantee.

\section{Preliminaries}
\label{preliminaries}
In this section, we illustrate the privacy definition, shuffle model and several properties of differential privacy, all of which are prepared for the proposed methods.
\subsection{Central and Local Differential Privacy}
Differential privacy (DP) \cite{dwork2014algorithmic} is a \emph{de facto} standard that is widely accepted to preserve privacy in FL. The notion is typically built up in a central setting where a trusted server can access the raw data. 
Local differential privacy (LDP) \cite{erlingsson2014rappor}, on the other hand, offers users a stronger privacy guarantee for the settings without assumption of trusted server.
\begin{Defs}[Differential Privacy]
	For any $\epsilon, \delta \geq 0$, a randomized algorithm $M: \mathcal{D} \rightarrow \mathcal{Z}$ is $(\epsilon, \delta)$-differential privacy if for any neighboring datasets $D, D' \in \mathcal{D}$ and any subsets $S \subseteq \mathcal{Z}$,
	\begin{footnotesize}
	\[
	\Pr[M(D) \in S] \leq e^{\epsilon}\Pr[M(D') \in S] + \delta
	\]
	\end{footnotesize}
\end{Defs}

\begin{Defs}[Local Differential Privacy]
	For any $\epsilon, \delta \geq 0$, an algorithm $M: \mathcal{D} \rightarrow \mathcal{Z}$ is $(\epsilon, \delta)$-local differential privacy if $\forall g, g' \in \mathcal{D}$ and $\forall z \in \mathcal{Z}$,
	\begin{footnotesize}
	\[
	\Pr[M(g) = z] \leq e^{\epsilon}\Pr[M(g') = z ] + \delta
	\]
	\end{footnotesize}
\end{Defs}

\begin{figure} 
  \centering
  \includegraphics[scale=0.15]{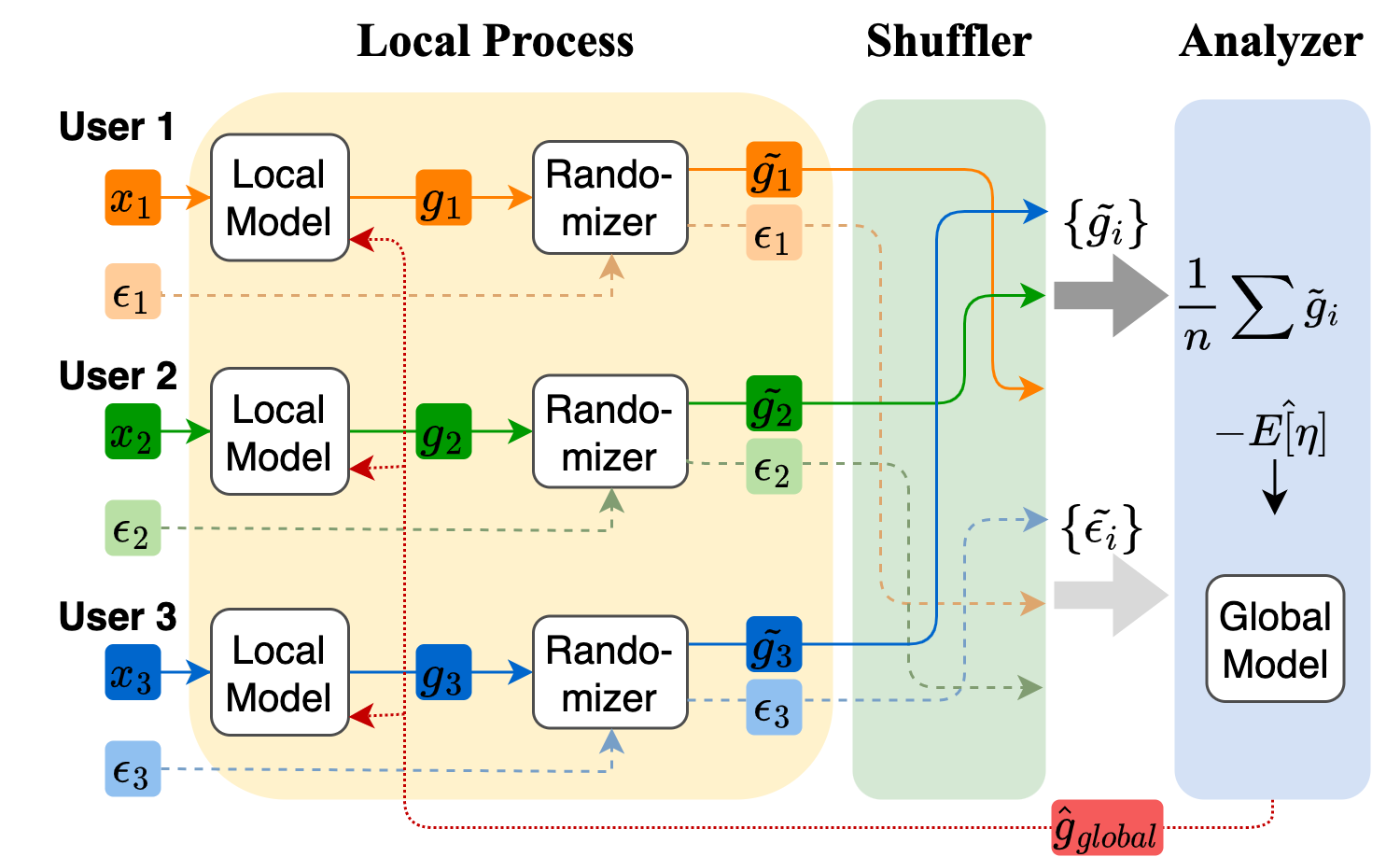}
  \caption{Procedure of APES. Gradients $g_i$ trained by user data $x_i$ are randomized locally, then privacy parameters $\epsilon_i$ and $g_i$ are shuffled separately. Analyzer acts as the curator to aggregate and calibrate gradients $\tilde{g}_i$ for global model. }
  \label{fig:PerS_framework}
\end{figure}

\subsection{Shuffle-based Privacy}
Shuffle model \cite{bittau2017prochlo} was proposed to strengthen central privacy while preserving local user privacy. 
Given $n$ datapoints as the dataset $D=\{g_1, g_2, ..., g_n\}$,  each $g_i \in D$ owned by user $i$ is perturbed locally by a randomizer $M: \mathcal{D} \rightarrow \mathcal{Z}$ to ensure $(\epsilon^l, \delta^l)$-LDP before being sent to shuffler. 
Shuffler, a trusted third party, permutes and releases all the datapoints by algorithm $S: Z \rightarrow Z$ to analyzer. Untrusted analyzer aggregates all the datapoints.
The process $P=S \circ M$ satisfies at least $(\epsilon^l, \delta^l)$-DP against analyzer (cf. Lemma \ref{parallel_comp}). Recent works \cite{erlingsson2019amplification, balle2019privacy, girgis2021renyi, feldman2022hiding} achieve a much stronger central privacy guarantee, which is considered as privacy amplification effect by shuffling.
Among existing works, \citeauthor{feldman2022hiding} provides a tight privacy upper bound for single-message summation. 
Take neighboring datasets $D$ and $D'$ that only differ at $g_1$ (or $g'_1$), any perturbed datapoint $\tilde{g}_i$ can be regarded as a sampling from the distribution of a specific perturbed point $\tilde{g}_1$ or $\tilde{g}'_1$ with probability $\exp(-\epsilon^l)$. By this observation, the privacy bound is yielded.

\subsection{Privacy Tools}
As a general technique to implement DP, Laplace Mechanism \cite{dwork2014algorithmic} perturbs numerical values.
\begin{Defs}[Laplace Mechanism]
Given any function $f: \mathcal{D} \rightarrow \mathcal{Z}^d$ and neighboring datasets $D$ and $D'$, let $\Delta f= \max||f(D)-f(D')||_1$ be the sensitivity function, Laplace mechanism $M(D) = f(D) + Y^d$ satisfies $\epsilon$-DP, where $Y^d$ is random variable $i.i.d$ drawn from distribution $Lap(0, \frac{\Delta f}{\epsilon})$.
\end{Defs}

Composition theorems provide tight bounds for the algorithm combined with several DP blocks. 
\begin{lemma}[Parallel Composition]\cite{yu2019differentially}
\label{parallel_comp}
	Given an $(\epsilon_i, \delta_i)$-DP algorithm $M_i: \mathcal{D} \rightarrow \mathcal{Z}$ for $i \in [m]$, a class of $\{M_i\}_{i\in[m]}$ on disjoint subsets of $D$ is $(\max{\epsilon_i}, \max{\delta_i})$-DP.
\end{lemma}


\begin{lemma}[Advanced Composition]\cite{dwork2014algorithmic}
\label{advanced_comp} Given an $(\epsilon_i, \delta_i)$-DP algorithm $M_i: \mathcal{D} \rightarrow \mathcal{Z}$ for $i \in [m]$, the sequence of $\{M_i\}_{i\in[m]}$ on the same dataset $D$ under m-fold composition is $(\epsilon', \delta'+m\delta)$-DP where $\epsilon'=\epsilon\sqrt{2m\log{1/\delta'}}+ m\epsilon(e^{\epsilon}-1)$.
\end{lemma}

No matter what dataset or query is adopted, any privacy mechanism can be reduced to a basic random response with the same privacy level \cite{kairouz2015composition}.
\begin{lemma}[Degraded Privacy]
\label{degraded_priv}
	For any $\epsilon$-DP mechanism $M$, for $X:\{x, \bar{x}\}$, $\exists \tilde{M}$ dominates $M$ where:
	\begin{footnotesize}
	\begin{equation*}
	 \Pr[\tilde{M}(x)=z]=\left\{
    \begin{array}{rcl}
    \frac{e^\epsilon}{1+e^\epsilon},     &  z=x\\
   \frac{1}{1+e^\epsilon},     & z=\bar{x} 
    \end{array}
    \right.
    \end{equation*}
    \end{footnotesize}
\end{lemma}

\section{Proposed Methods}
This section illustrates our methods for a strong privacy amplification effect.
We first introduce Clip-Laplace Mechanism to implement the effect. Then two frameworks are proposed. APES is a general framework which shuffles both privacy parameters and gradients, the improved S-APES sparsifies dimensions without downgrading shuffling effect.



\subsection{Clip-Laplace Mechanism}

To make bounding privacy while maintaining model accuracy possible, it is necessary to introduce a mechanism for LDP with a finite and fixed output range. Existing works on this task provide non-fixed output ranges \cite{geng2018truncated}, or increase noise scale when ranges of input and output are not overlapped \cite{holohan2018bounded, croft2022differential}. 
To address this issue, we introduce a variant of Laplace Mechanism, \emph{Clip-Laplace}, 
which provides $\epsilon$-DP for continuous real values with the same finite output ranges.

\begin{Defs}[Clip-Laplace Mechanism]
\label{def:clip-lap}
	Given any function $f: \mathcal{X} \rightarrow \mathcal{Y}^d$ and sensitivity $\Delta f = \max||f(X)-f(X')||_1$ for any neighboring datasets $X$ and $X'$. Clip-Laplace Mechanism is $M:\mathcal{Y}^d  \rightarrow \mathcal{Z}^d$. Each $Z \in \mathcal{Z}^d$ is a r.v. i.i.d. drawn from distribution $CLap(f(x), \lambda, A)$ of which the probability density function is defined as follows:
	\begin{footnotesize}
	\[
    p(z)=\left\{
    \begin{array}{rcl}
    \frac{1}{2\lambda S}\exp{(-\frac{|z-f(x)|}{\lambda})},     &  -A\le z \le A\\
    0,     & otherwise 
    \end{array}
    \right.
    \]
    \end{footnotesize}
    where normalization factor $S=1-\frac{1}{2}\exp(\frac{-A+f(x)}{\lambda})-\frac{1}{2}\exp(\frac{-A-f(x)}{\lambda})$ and $A \geq \Delta f /2$.
\end{Defs}

\begin{theorem}
\label{clip_lap_DP}
Clip Laplace mechanism preserves $\epsilon$-LDP when the $f(x) \in [-\Delta f/2, \Delta f/2]$, and $\lambda=\Delta f / \epsilon$.
\end{theorem}
The proof is provided in Appendix A.

\emph{Discussion.} (\romannumeral1) When achieving the same level of $\epsilon$-LDP,  the variance of Clip-Laplacian outputs is smaller than classic Laplacian outputs. This property is based on the assumption of symmetric limited range of inputs (cf. Theorem \ref{clip_lap_DP}), which is reasonable for many fields such as gradients aggregation, location statistics, financial analysis and so on. 
(\romannumeral2) The Clip-Laplacian outputs are biased. A feasible solution for correction is to calibrate the outputs with the expectation, which can be estimated when privacy parameters are given.

\subsection{APES Framework}
\label{sec_perS}
We formalize \textbf{APES}, a privacy \textbf{\underline{A}}mplification framework for  \textbf{\underline{PE}}rsonalized private federated learning with \textbf{\underline{S}}huffle Model.
The framework includes three procedures: local updating, shuffling and analyzing process with three parties separately. Convergence upper bound of APES is given at last.

\subsubsection{Architecture}
Consider 3 parties: (\romannumeral1) $n$ users, each holds a dataset $X_i$ and a randomizer $M_i$ satisfying $\epsilon^l_i$-LDP. (\romannumeral2) A shuffler with algorithm $S$. (\romannumeral3) An analyzer, trains global model with shuffled messages.
The process $P = S \circ M$ ensures $(\epsilon^c, \epsilon^c)$-DP for global model, where $M = (M_1, ..., M_n)$ with $\epsilon^l=(\epsilon^l_1,  ..., \epsilon^l_n)$ for dimension level.

\subsubsection{Basic Framework}
Algorithm \ref{algo_basic} outlines the procedures of APES. We denote clip bound by $C$, learning rate by $\alpha$ and training epochs by $T$. Main procedures are as follows:
\begin{itemize}
	\item \emph{Local Updating}. 
	Each user randomizes each dimension of model gradient $g_i$ with  $\epsilon^l_i$ by applying Clip-Laplace Mechanism. 
	Both perturbed gradient $\tilde{g}_i$ and $\epsilon_i^l$ are sent to Shuffler. To keep the order of dimensions, dimension index $k$ of $\tilde{g}_i$ is sent as well.
	\item \emph{Shuffling Process}. Shuffler shuffles $\{\tilde{g}_i\}_{i \in [n]}$  within the same dimension, $\{\epsilon_i^l\}_{i \in [n]}$ is also permuted.
	\item \emph{Analyzing Process}. 
 Considering Clip-Laplace Mechanism is biased, the average gradient $\tilde{g}$ needs to be calibrated.
We cannot calibrate $\tilde{g}_i$ one by one as the correspondence of $\epsilon^l_i$ and $g_i$ is invisible to analyzer.  
 Empirically, we observe that the value of $\tilde{g}$ is close to the value of $\mathbb{E}[\tilde{\bar{g}}]$  (cf. Fig. 4 in Appendix C), where $\bar{g}=\frac{1}{n}\sum_{i=1}^n{g_i}$, $\mathbb{E}[\tilde{\bar{g}}] = \frac{1}{n}\sum_{i=1}^n{\mathbb{E}[\tilde{\bar{g}}_i]}$ and $\tilde{\bar{g}}_i \sim CLap(\bar{g}, 2C/\epsilon^l_i, C)$. Hence we can estimate the clean gradients $\bar{g}$ by approximating $\mathbb{E}[\tilde{g}]$ with $\mathbb{E}[\tilde{\bar{g}}]$. Specifically, $\mathbb{E}[\tilde{g}]$ is estimated by $\tilde{g}$, and each term of $E[\tilde{\bar{g}}]$ with $\epsilon^l_i$ is as follows:
\begin{footnotesize}
\begin{equation}
	\label{cali_noise}
	\mathbb{E}[\tilde{\bar{g}}_i]
	= \frac{(C+\lambda_i) \cdot (e_1 - e_2) + 2\bar{g}}{2 - e_1 - e_2}
\end{equation}
\end{footnotesize}
where $e_1=e^{\frac{-C-\bar{g}}{\lambda_i}}$, $e_2=e^{\frac{-C+\bar{g}}{\lambda_i}}$ and $\lambda_i = 2C/\epsilon^l_i$.
\end{itemize}

\renewcommand{\algorithmicrequire}{\textbf{Input}}
\renewcommand{\algorithmicensure}{\textbf{Output}}
\algnewcommand{\LineComment}[1]{\State \(\triangleright\) #1}
\begin{algorithm}[tb]
	\caption{Basic Framework: APES}
	\label{algo_basic}
	\begin{algorithmic}
	\Require $T$, $\{(X_i, \epsilon^l_i) \}_{i \in [n]}$, $h(w)$, $C$, $\alpha$.
		\Ensure model $w$
		\State Analyzer initializes and broadcasts $w^{(0)}$.
		\For{$t=1,2,...,T$}
			\LineComment{\texttt{Local Updating}}
			\For{each user $i \in [n]$} 
				\State $w_i \leftarrow w^{(t)}$ \Comment{Update local model}
				\State $g_i \leftarrow \nabla_{w_i} h(w_i, X_i)$
				\State $\bar{g_i} \leftarrow {\rm Clip} (g_i, -C, C)$
				\State $\tilde{g_i} \leftarrow  {Randomize(\cdot)}$ \Comment{Local perturbation}
				\State user $i$ uploads $(\tilde{g_i}, \epsilon^l_i)$ to Shuffler
			\EndFor
			\LineComment{\texttt{Shuffling Process}}
			\For{each dimension $k \in [d]$}
				\State generate permutation $\pi_k$ over $[d]$
				\State $\{(\tilde{g}_{i,\pi_k(k)},k)\}_{i \in [n]} \leftarrow {\rm Shuffle}(\pi_k, \{\tilde{g}_{i,k}\}_{i\in [n]})$ 
			\EndFor
			\State generate permutation $\pi$ over $[n]$
			\State $\{\epsilon_{\pi(i)}^l\}_{i\in [n]} \leftarrow {\rm Shuffle}(\pi, \{\epsilon_i^l\}_{i\in [n]})$ \Comment{Shuffle $\epsilon$}
			\State send $\{\{(\tilde{g}_{i,\pi_k(k)},k)\}_{i \in [n]}\}_{k\in [d]}$ and $\{\epsilon_{\pi(i)}^l\}_{i\in [n]}$ 
			\LineComment{\texttt{Analyzing Process}}
			\For{each dimension $k \in [d]$}
				\State $\tilde{g}_k \leftarrow \frac{1}{n}\sum_i \tilde{g}_{i,k}$ \Comment{Aggregate by dimension}
                \EndFor
				\State $\hat{g} \leftarrow {\rm Calibrate}(\tilde{g}, \{\epsilon_i\}_{i \in [n]})$ 
			\State $w^{(t+1)} \leftarrow w^{(t)} - \alpha \hat{g}$ and broadcast.
		\EndFor
		\Return $w^{(T)}$
	\end{algorithmic}
\end{algorithm}

\subsubsection{Convergence Analysis}
To demonstrate the performance of global model under Clip-Laplace perturbation, we provide the upper bound of convergence of Algorithm \ref{algo_basic} with the objective function $h(w;w^{(0)})=F(w)+\frac{\mu}{2}||w-w^{(0)}||^2$. The regularization term $\frac{\mu}{2}||w-w^{(0)}||^2$ of $h$ is introduced for the ease of calculation \cite{li2020federated}.
\begin{theorem}[Convergence Upper Bound]
\label{th_conv}
After T aggregations, the expected decrease in the global loss function $f(w) = \frac{1}{n}\sum_iF_i(w)$ of APES is bounded as follows:
\begin{footnotesize}
\begin{align*}
   & \mathbb{E}[f(\tilde{w}^{(T)})-f(w^*)]
    \le a_1^T(\mathbb{E}[f(\tilde{w}^{(0)})]-f(w^*)) \\
    &+ \frac{a_1^T-1}{a_1-1}(O(a_2C/\min(\epsilon^l_i)) + O(a_3 C^2/\min(\epsilon^l_i)^2))
\end{align*}
\end{footnotesize}
where $a_1=1+\frac{2\beta(\alpha B-1)}{\mu} + \frac{2\beta LB(\alpha+1)}{\mu\bar{\mu}} + \frac{2\beta LB^2(1+\alpha)^2}{\bar{\mu}^2}, a_2=L(\frac{1}{\mu} + \frac{BL(1+\alpha)}{\bar{\mu}}), a_3=\frac{L}{2}$.
\end{theorem}
The proof refers to Appendix B.

\emph{Discussion.} The convergence upper bound increases as the bias and variance (the second and the third term) of Clip-Laplace perturbation grow, of which the influence is the same as classic Laplace Mechanism.



\subsection{S-APES Framework}
\label{sec_PerSS}
To strengthen privacy in the high-dimension scenario, we propose \textbf{S-APES} framework, which improves \textbf{\underline{APES}} with post-\textbf{\underline{S}}parsification technique.

Since gradients are usually high-dimensional, limiting the number of dimensions helps to save the privacy cost \cite{ye2020local, duan2022utility}. Selecting part of dimensions with large magnitude keeps majority of information \cite{aji2017sparse} and reduces privacy loss, but needs extra protection since the selection itself is data-dependent process. 
To select informative dimensions without breaching privacy, we propose post-parsification technique.

\subsubsection{Post Sparsification}
Algorithm \ref{algo_improved} demonstrates the local process of S-APES with post-sparsification.
Concretely, each user $i$ is asked to select the largest $b$ absolute values over $d$ dimensions of $\tilde{g}_i$. To keep the selected dimension index private, the selection is executed after local perturbation. For avoiding the shuffling effect degradation caused by members reduction, each user pads the rest of $(d-b)$ dimensions with perturbed $0$. 
Denote sparsification process by $K$, the whole process of S-APES is defined as $P_s=S \circ K \circ M$.

\begin{algorithm}[tb]
	\caption{$Randomize(\cdot)$ for S-APES}
	\label{algo_improved}
	\begin{algorithmic}
		\Require $\{(g_i, \epsilon^l_i) \}_{i \in [n]}$, $C$.
		\Ensure perturbed gradient $\tilde{g}_i$
		\State $\tilde{g}_{i} \leftarrow {\rm CLap}(0, (d\Delta f)/\epsilon_i^l, C)$ \Comment{Clip-Laplace perturbing}
		\State $I_b \leftarrow  \{k | k \in max({|\tilde{g}_{i,k}|}_{k \in [d]})\}^{b}$ \Comment{Post-top-b index set}
		\For{each index $k \notin I_b$}
			\State $\tilde{g}_{i,k} \leftarrow {\rm CLap}(0, (d\Delta f)/\epsilon_i^l, C)$ \Comment{Dummy padding}
		\EndFor
	\Return $\tilde{g_i}$
	\end{algorithmic}
\end{algorithm}

\section{Privacy Analysis}
\label{sec:privacy_analysis}
In this section, we first derive a na\"ive privacy bound based on existing works, then show the local and central privacy bound of our frameworks. The sketch of privacy amplification effect analysis is provided at last.
\subsection{Baseline Results}
To analyze the privacy amplification effect of shuffling under personalized LDP,
the most na\"ive way is applying existing shuffling bounds \cite{erlingsson2019amplification, balle2019privacy, girgis2021renyi, feldman2022hiding} on heterogeneous local privacy budgets, i.e., $\epsilon^l_i$, with classic Laplace Mechanism. \textcolor{black}{However, different $\epsilon^l_i$ lead to different scales of the Laplacian distributions and their divergence may be infinite. As a result, the central privacy may be unbounded. Hence based on the previous work \cite{feldman2022hiding} we can only approximate the true bound by using the same maximum $\epsilon^l_i$ for all users: }

	\begin{footnotesize}
	\begin{equation}
	\label{eq_FMT_max}
	\epsilon^c \leq \ln (1+ \frac{e^{\max(\epsilon_i^l)} - 1}{e^{\max(\epsilon_i^l)} + 1}(\frac{8(e^{\max(\epsilon_i^l)}\log(4/\delta))^{1/2}}{n^{1/2}} + \frac{8e^{\max(\epsilon^l_i)}}{n}))
	\end{equation}
	\end{footnotesize}
\subsection{Main Results}
Proposed techniques Clip-Laplace Mechanism and Neighbor Divergence make analyzing privacy amplification effect possible. 
Without loss of generality, we suppose two neighboring datasets $D=\{g_1, g_2, ..., g_n\}$ and $D'=\{g'_1, g_2, ..., g_n\}$ that only differs at $g_1$ or $g'_1$ of user $1$, and provide privacy bounds of our frameworks as follows.
\begin{theorem}[Local Bound] 
\label{th_local_bound}
Given $\epsilon^l=(\epsilon^l_1, ..., \epsilon^l_n)$, the local process $M=(M_1, ..., M_n)$ of APES on $d$-dimension gradients satisfies $\epsilon^l_i$-LDP in dimension level, $d\epsilon^l_i$-LDP in user level for each user $i$.
\end{theorem}
\emph{Discussion.} Our frameworks achieve personalized LDP for each user. This comes from Theorem \ref{clip_lap_DP}.

\begin{theorem}[Central Upper bound]
\label{th_upper_bound}
	Let $i,j\in[n]$, $\delta_s \in [0,1]$, $\sum_{i=2}^n\sum_{j=1}^n \frac{p_{ij}}{n} \geq 16\ln(4/\delta_s)$, 
$P=S \circ M$ of APES satisfies $(\epsilon^c, \delta^c)$-DP where $\delta^c \leq \frac{e^{\epsilon^l_j}-1}{e^{\epsilon^l_j}+1}\delta_s$,\\
	\begin{footnotesize}
	\begin{equation*}
		\epsilon^c \leq  \ln(1+ \frac{e^{\max(\epsilon^l_j)}-1}{e^{\max(\epsilon^l_j)}+1}(\frac{8(\ln(4/\delta_s))^{1/2}}{(\sum\limits_{i=2}^n\sum\limits_{j=1}^n \frac{p_{ij}}{n})^{1/2}} + \frac{8}{\sum\limits_{i=2}^n\sum\limits_{j=1}^n \frac{p_{ij}}{n}}))
	\end{equation*}
	\end{footnotesize}
	when $\sum_{i=2}^n\sum_{j=1}^n \frac{p_{ij}}{n} \geq 16\ln(4/\delta_s)$, $\delta_s \in [0,1]$ and $p_{ij}=\frac{\epsilon^l_i}{\epsilon^l_j} \cdot \frac{1-e^{-\epsilon^l_j}}{1-e^{-\epsilon^l_i}} \cdot e^{-\max(\epsilon^l_i, \epsilon^l_j)}$.
\end{theorem}
\emph{Discussion.} APES gains a strong central privacy for dimension level. Theorem \ref{th_upper_bound} indicates most users are provided with a much stricter central privacy as $\epsilon^c$ than their local privacy $\epsilon^l_i$. 
A sketch of the proof is provided in the following section.

\begin{proposition}[User Level Central Bound]
\label{user_level_priv}
With ${\delta'}^{uc}>0$ and $0<b\leq d$, the process $P_s=S \circ K \circ M$ of S-APES with $b$-dimension sparsification is $(\epsilon^{uc}, \delta^{uc})$-DP where
$\epsilon^{uc} = \epsilon^c \sqrt{4b\ln(1/\delta^{uc})} + 2b\epsilon^c(\exp{(\epsilon^c)}-1)$ and 
$\delta^{uc} = {\delta'}^{{uc}} + 2b\delta^c$.
\end{proposition}
\emph{Discussion.} S-APES achieves the same dimension-level $\epsilon^c$ as APES. Considering dimensions of a gradient are not independent and extracting $b$ dimensions leads to $2b$ sensitivity, we derive the user-level privacy amplification effect by composition theorems. Note that $\epsilon^{uc}$ grows linearly with $b$, which implies privacy loss reduces when fewer dimensions are uploaded by post-sparsification.

\subsection{EoN: Privacy Amplification Analysis}
\label{section:EoN}
To analyze privacy of proposed frameworks, we first introduce Neighbor Divergence, then present the sketch of Echo of Neighbors (EoN) analysis for privacy amplification effect.

\subsubsection{Neighbor Divergence}
We introduce \emph{Neighbor Divergence} to characterize how well a user's output distribution closes the gap between itself and other users' distributions. Concretely, it defines the distance among output distributions of local randomizers of users with heterogeneous privacy budgets and different raw datapoints. 
\begin{Defs}[Neighbor Divergence] 
\label{def_ND}
Consider any $g_s$, $g_t \in \mathcal{D}$ and randomizers $M_i$, $M_j$ satisfying $\epsilon^l_i$, $\epsilon^l_j$-LDP separately. Let $\mu^{(s)}_i$ and $\mu^{(t)}_j$ be distributions of $M_i(g_s)$ and $M_j(g_t)$ respectively, $U^{(s)}_i \sim \mu^{(s)}_i$, $U^{(t)}_j \sim \mu^{(t)}_j$, the neighbor divergence between $\mu^{(s)}_i$ and $\mu^{(t)}_j$ is defined as:
	\begin{footnotesize}
	\[
	D_{N}(\mu^{(s)}_i||\mu^{(t)}_j) = \max\limits_{S \subseteq Supp(U^{(s)}_i)} [\ln \frac{\Pr[U^{(s)}_i \in S]}{\Pr[U^{(t)}_j \in S]}]
	\]
	\end{footnotesize}
\end{Defs}
In particular, the neighbor divergence under Clip-Laplace Mechanism is demonstrated as follows.
\begin{lemma}
\label{lemma_nd_of_clap}
	Let $f(x) \in [-C, C]$, $\lambda=\Delta f /\epsilon^l$ and $\Delta f = 2C$, the neighbor divergence between distribution $\mu^{(s)}_i$ and $\mu^{(t)}_j$ under Clip-Laplace Mechanism is $D_{N}(\mu^{(s)}_i||\mu^{(t)}_j) \leq \ln(\alpha \frac{\epsilon^l_i}{\epsilon^l_j} e^{(\frac{(\epsilon^l_i+\epsilon^l_j)}{2} + \frac{A|\epsilon^l_i-\epsilon^l_j|}{2C})})$. Specifically,  $D_{N}(\mu^{(s)}_i||\mu^{(t)}_j) \leq \ln(\frac{\epsilon^l_i}{\epsilon^l_j} \cdot \frac{1-e^{-\epsilon^l_j}}{1-e^{-\epsilon^l_i}} \cdot e^{\max(\epsilon^l_i, \epsilon^l_j)})$ when $A=C$.  $\alpha$ denotes $\frac{({1-\frac{1}{2}\exp(\frac{{\epsilon^l_j(-A+C)}}{2C})-\frac{1}{2}\exp(\frac{\epsilon^l_j(-A-C)}{2C})})}{(1-\frac{1}{2}\exp(\frac{{\epsilon^l_i(-A+C)}}{2C})-\frac{1}{2}\exp(\frac{\epsilon^l_i(-A-C)}{2C}))}$.
\end{lemma}

\subsubsection{A sketch of EoN Analysis}
We analyze the central privacy bound in Theorem \ref{th_upper_bound} with the observation of \emph{Echos of Neighbors}.
There are three main steps:
(\romannumeral1) After shuffling, output distributions of the rest users are converted into the same distribution of user $1$ which can be seen as ``echos'' by neighbor divergence.
(\romannumeral2) Then all the ``echos'' are transformed into certain distributions which disentangle from different $\epsilon^l_i$ by degraded privacy. These distributions form a mixed distribution.
(\romannumeral3) Finally, we measure the divergence between the mixed distributions on $D$ and $D'$.

\textbf{Step (\romannumeral1).} Recall that LDP mechanism $M_i: \mathcal{Y} \rightarrow \mathcal{Z}$ satisfying $\epsilon^l_i$-LDP for any $i \in [n]$. Based on neighbor divergence, for any $\mu_i^{(s)}$ and $\mu_j^{(t)}$ we have $p_{ij} \leq {\mu_i^{(s)}}/{\mu_j^{(t)}} \leq e^{-D_N(\mu_j^{(t)}||\mu^{(s)}_i)}$ by Definition \ref{def_ND}.
Specifically, for any user's distribution $\mu_i^{(i)}$ on $g_i \in D\backslash\{g_1, g'_1\}$, ``echo'' $\mu_j^{(1)}$(or ${\mu'}_j^{(1)}$) of user $1$ with $g_1$ (or $g'_1$) is generated as follows:

\begin{footnotesize}
\begin{equation}
	\mu_i^{(i)} = \frac{p_{ij}}{2}\mu_j^{(1)} + \frac{p_{ij}}{2}{\mu'}_j^{(1)} + (1-p_{ij})\gamma_i^{(i)}
\end{equation}
\end{footnotesize}
The distribution $\gamma_i^{(i)} = \mu_i^{(i)}-p_{ij}/2\cdot (\mu_j^{(1)}+{\mu'}_j^{(1)}) / (1-p_{ij})$. The idea is inspired by a prior work \cite{feldman2022hiding}. Consider the situation that both $g$ and $\epsilon^l$ are shuffled, the correspondence between $g_i$ and $\epsilon^l_i$ is broken. 
An adversary cannot decide which $\epsilon^l_j$ is used for perturbing $g_1$, hence any value in $\{\epsilon^l_i\}$ is possible. Based on it we derive a general bound, then consider the worst-case situation with the largest $\epsilon^l_j$ on $g_1$ for the upper bound at step (\romannumeral3).

\textbf{Step (\romannumeral2).} Except for user $1$, the mixed distribution of multiple $\mu_j^{(1)}$ with different $\epsilon^l_j$ from $n-1$ users is still hard to bound. Hence, with the help of degraded privacy (cf. Lemma \ref{degraded_priv})  we transform $(\mu_j^{(1)} + {\mu'}_j^{(1)})$ into $(\rho^{(1)} + {\rho'}^{(1)})$ for any $j \in [n]$, then $\epsilon^l_j$ is disentangled from $\mu_j^{(1)}$.
\begin{lemma}[Transformation] 
\label{th_transformation}
Let $\rho^{(1)}$ and ${\rho'}^{(1)}$ denote the distribution of function $\tilde{M}:\{g_1, g'_1\} \rightarrow \mathcal{Z}$,
$\mu_i^{(i)}$ be the distribution of $M_i(g_i)$, and $\gamma_i^{(i)}$ be the rest part of $\mu_i^{(i)}$ except ${\rho}^{(1)}$ and ${\rho'}^{(1)}$, then $\mu_i^{(i)}$ is mapped as follows:
	\label{th_EoN_2}
	\begin{footnotesize}
	\begin{equation}
		\label{eq_EoN_2}
		\mu_i^{(i)} = \frac{1}{n}\sum_{j=1}^n(\frac{p_{ij}}{2}\rho^{(1)} + \frac{p_{ij}}{2}{\rho'}^{(1)} + (1-p_{ij})\gamma_i^{(i)})
	\end{equation}
	\end{footnotesize}
		where $p_{ij}=\exp(-D_N(\mu^{(1)}_j||\mu^{(i)}_i))$.
\end{lemma}
\begin{proof}
	By Lemma \ref{degraded_priv}, we have $\mu_j^{(1)} = ({e^{\epsilon^l_j}}/({1+e^{\epsilon^l_j}}))\rho^{(1)} + ({1}/({1+e^{\epsilon^l_j}})){\rho'}^{(1)}$ and ${\mu'}_j^{(1)} = ({1}/({1+e^{\epsilon^l_j}}))\rho^{(1)} + ({e^{\epsilon^l_j}}/({1+e^{\epsilon^l_j}})){\rho'}^{(1)}$. The influence of $\epsilon^l_j$ on $\mu_i^{(i)}$ is isolated:\\
\begin{footnotesize}
\begin{align*}
	&\mu_i^{(i)} = \frac{1}{n}\sum_{j=1}^n(\frac{p_{ij}}{2}\mu_j^{(1)} + \frac{p_{ij}}{2}{\mu'}_j^{(1)} + (1-p_{ij})\gamma_i^{(i)}) \\
	&= \frac{1}{n}\sum_{j=1}^n(\frac{p_{ij}}{2}\rho^{(1)} + \frac{p_{ij}}{2}{\rho'}^{(1)} + (1-p_{ij})\gamma_i^{(i)})
\end{align*}
\end{footnotesize}
\end{proof}

\textbf{Step (\romannumeral3).} Now we can bound the divergence of the transformed distributions on $D$ and $D'$.
\begin{lemma}[Generalized Central Bound]
\label{lemma:general_privacy_bound}
	Let $i,j\in[n]$, $\delta_s \in [0,1]$, $\sum_{i=2}^n\sum_{j=1}^n \frac{p_{ij}}{n} \geq 16\ln(4/\delta_s)$, $P=S \circ M$ of APES on $D$ and $D'$ is $(\epsilon^c, \delta^c)$-distinguishable where $\delta^c \leq \frac{e^{\epsilon^*}-1}{e^{\epsilon^*}+1}\delta_s$ and $p_{ij}=\frac{\epsilon_i^l}{\epsilon_j^l}\cdot \frac{1-e^{-\epsilon_j^l}}{1-e^{-\epsilon_i^l}}\cdot e^{-\max(\epsilon^l_i, \epsilon^l_j)}$,\\
	\begin{footnotesize}
	\begin{equation*}
		\epsilon^c \leq \ln(1+ \frac{e^{\epsilon^*}-1}{e^{\epsilon^*}+1}(\frac{8(\ln(4/\delta_s))^{1/2}}{(\sum\limits_{i=2}^n\sum\limits_{j=1}^n \frac{p_{ij}}{n})^{1/2}} + \frac{8}{\sum\limits_{i=2}^n\sum\limits_{j=1}^n \frac{p_{ij}}{n}}))
	\end{equation*}
	\end{footnotesize}
\end{lemma}

\begin{proof}
By Lemma \ref{th_transformation}, any output distribution $\mu_i^{(i)}$ can be mapped into $\rho^{(1)}$ or ${\rho'}^{(1)}$ with probability $p_{ij}/2n$, into $\gamma_i^{(i)}$ with $(1-p_{ij})/n$. Consider outputs of $n-1$ users, we get a set of mapping distributions including $n(n-1)$ elements.

With any $T \subseteq [n(n-1)]$, $\Gamma = [n(n-1)]\backslash T$, we define an mapping event $U = \{u_1, ... u_{n(n-1)} \}$ where 
\begin{footnotesize}
\begin{equation*}
	u_t = \left\{
\begin{array}{rcl}
	\rho^{(1)} \text{\rm or } {\rho'}^{(1)},     &  t \in T \\
	\gamma_t^{(t)}, & t \in \Gamma
\end{array}
\right.
\end{equation*}
\end{footnotesize}
The effect of $\gamma_i$ can be removed in process $P$ under the same $U_T$ since all the $u_t \in U_{\Gamma}$ are the same in $D$ and $D'$:
\begin{footnotesize}
\begin{align}
\label{eq_gamma_remov}
	\frac{\Pr[P(D)=\mathbf{z}]}{\Pr[P(D')=\mathbf{z}]}
	\leq \frac{\Pr[U_T\cup{\rho}^{(1)}]\Pr[U_\Gamma]}{\Pr[U_T\cup{\rho'}^{(1)}]\Pr[U_\Gamma]}
\end{align}
\end{footnotesize}

Then we define $T_0 \subseteq T$ and $T_1=T \backslash T_0$, $\forall u_t \in U_{T_0}$ is $\rho^{(1)}$,  $\forall u_t \in U_{T_1}$ is ${\rho'}^{(1)}$ on $D$; $T'_0 \subseteq T$ and $T'_1=T \backslash T'_0$, $\forall u_t \in U_{T'_0}$ is $\rho^{(1)}$,  $\forall u_t \in U_{T'_1}$ is ${\rho'}^{(1)}$ on $D'$. Put aside the randomness on $g_1$ and $g'_1$ for now (which means the output of user $1$ can be regarded as ${\rho}^{(1)}$ or ${\rho'}^{(1)}$), when reaching the mixed output $\mathbf{z}$ with the same number of $\rho^{(1)}$ or ${\rho'}^{(1)}$, $U_{T_0}$ on $D$ and $U_{T'_0}$ on $D'$ should be different as $|T '_0|-|T_0|=1$.
Recall that $|T| \sim \sum_{i=2}^n\sum_{j=1}^n {\rm Bern}(p_{ij}/n)$ and $|T_0| \sim {\rm Bin}(1/2, |T|)$ according to Lemma \ref{th_transformation}, we can bound Eq.\eqref{eq_gamma_remov} by deriving following equation:
\begin{footnotesize}
\begin{align}
\label{eq_general1}
	&\frac{\Pr[U_T\cup{\rho}^{(1)}]}{\Pr[U_T\cup{\rho'}^{(1)}]} 
	= \frac{\Pr[U_{T_0}\cup U_{T_1} | U_T] \cdot \Pr[U_T]}{\Pr[U_{T'_0}\cup U_{T'_1}  | U_T] \cdot \Pr[U_T]} \nonumber \\
	& = \frac{\tbinom{|T|}{|T_0|} (\frac{1}{2})^{|T_0|} (\frac{1}{2})^{|T|-|T_0|}}{\tbinom{|T|}{|T'_0|} {(\frac{1}{2})}^{|T'_0|} {(\frac{1}{2})}^{|T|-|T'_0|}}
	= \frac{|T_0|+1}{|T|-|T_0|}
\end{align}
\end{footnotesize}
With Chernoff bound and Hoeffding's inequality,
when $\sum_{i=2}^n\sum_{j=1}^n \frac{p_{ij}}{n} \geq 16\ln(4/\delta_s)$, Eq.\eqref{eq_general1} is bounded as $\frac{|T_0|+1}{|T|-|T_0|} \leq \ln(1+\frac{8(\ln(4/\delta_s))^{1/2}}{(\sum_{i=2}^n\sum_{j=1}^n \frac{p_{ij}}{n})^{1/2}} + \frac{8}{\sum_{i=2}^n\sum_{j=1}^n \frac{p_{ij}}{n}})$.

At last, we consider the randomness on $g_1$ and $g'_1$ with certain privacy budget $\epsilon^*$, the rest of the proof follows existing work \cite{feldman2022hiding} and the general bound is proved. The full proof of Lemma \ref{lemma:general_privacy_bound} is provided to Appendix A.
\end{proof}

From the analysis above, it is realized that which $\epsilon^*$ adopted by $g_1$ or $g_1'$ is crucial for the bound. For the worst case that $\epsilon^*=\max(\epsilon^l_j)$ for ${j\in[n]}$, the divergence is upper bounded as Theorem \ref{th_upper_bound}. 
The proof refers to Appendix A.

\section{Experiments}
\label{exp_sec}
We conduct comprehensive experiments on APES and S-APES with the public dataset and various privacy settings. 
\subsection{Experiment Settings}
\paragraph{Dataset and Implementation} QMNIST \cite{yadav2019cold} is an extended version of MNIST dataset \cite{lecun1998gradient}, which consists of 120,000 28-by-28-pixel images with 10 classes. We set users as $n$=10,000 and partition the dataset evenly for users. The frameworks are evaluated with a Logistic Regression model with $d$=7850. 
All the experiments are implemented on a workstation with an Intel(R) Xeon(R) E5-2640 v4 CPU at 2.40GHz and a NVIDIA Tesla P40 GPU running on Ubuntu. 
\paragraph{Baselines} We compare the proposed methods with the following schemes.
(\romannumeral1) Baseline frameworks include:
        \textbf{Non-Private}: FedProx \cite{li2020federated} without privacy protection.
         \textbf{LDP-Min}: all users adopt $\min \epsilon^l_i$ as privacy budget compulsively, which preserves privacy of all the users. 
	\textbf{PLDP}: FedProx with personalized LDP.
	\textbf{UniS}: FedProx with shuffle model under personalized LDP, the privacy bound refers to Eq. \eqref{eq_FMT_max}. All the baseline frameworks exploit classic Laplace Mechanism as local randomizer.
(\romannumeral2)  Baseline bounds of privacy amplification effect include: the numerical generic result of \textbf{BBGN'19} \cite{balle2019privacy}, the nemurical result of \textbf{FMT'22} \cite{feldman2022hiding}, the upper bound of \textbf{GDDTK'21} \cite{girgis2021renyi} and \textbf{Erlingsson'19} \cite{erlingsson2019amplification}.

\paragraph{Parameter Selection}
We stimulate the personalized privacy preference $\epsilon^l$ for several situations as Tab. \ref{table_eps_dist}. 
$\delta^s$ for shuffling is set to $10^{-8}$ and $\delta^{uc}$ after dimension composition is $3.6\times 10^{-5}$, smaller than $1/n$.



\subsection{Experiment Results}
We first show the effectiveness of the total frameworks, then confirm the privacy amplification effect, Clip-Laplace, and post-sparsification adopted in frameworks separately.
\paragraph{Effectiveness of frameworks} 
Tab. \ref{table_eps_amp} demonstrates that our frameworks achieve stronger central privacy with comparable or higher utility under the same personalized LDP. we compare the model accuracy and central privacy budgets of one epoch under Uniform2.
(\romannumeral1) APES gains stricter privacy and the highest accuracy than baseline private frameworks. Dimension-level $\epsilon^c$ and user-level $\epsilon^{uc}$ reduce by more than 21\% compared to UniS and PLDP. LDP-min gets tighter bound, yet the model performs poorly. There is no baseline framework achieves both better sides. The performance of APES benefits from privacy amplification effect of EoN and Clip-Laplace perturbation. 
(\romannumeral2) S-APES provides the same $\epsilon^c$ as APES and further enhances user-level privacy. $\epsilon^{uc}$ diminishes by 55.6\%, 66.7\%, 99.6\% compared to APES, UniS, and PLDP separately, It is noticed that local $\epsilon^{ul}$ also drops by dimension reduction. Though S-APES sacrifices accuracy of APES by 1.8\%, it is still higher than baselines. The post-sparsification in S-APES substantially boosts privacy with this tolerable utility reduction.

Fig. \ref{fig:priv_utility_b} confirms that our frameworks perform well on multiple distributions and ranges of $\epsilon^l$ locally (cf. Tab. \ref{table_eps_dist}).
(\romannumeral1) Accuracy of APES and S-APES is higher than UniS for the most settings. An exception is in Gauss1 LDP, which implies that S-APES may bot be appropriate for small $\epsilon^l$. Too much perturbation strengthens the privacy, but makes selecting informative dimensions harder.
(\romannumeral2) APES performs more stable than UniS for different $\epsilon^l$. A reasonable deduction is outputs of Clip-Laplace is not as sensitive as Laplace to varying parameters, which is verified in Fig. 7 in Appendix C.

\begin{table}
  \centering
  \begin{footnotesize}
  \begin{tabular}{lll}

    \toprule
    Name     & Distribution of $\epsilon^l=(\epsilon^l_1, ..., \epsilon^l_n)$ & Clip range    \\
    \midrule
    Uniform1 & $\mathcal{U}(0.05, 0.5)$ & $[0.05, 0.5]$ \\
    Uniform2 & $\mathcal{U}(0.05, 1)$ & $[0.05, 1]$ \\
    Gauss1 & $\mathcal{N}(0.1, 1)$ & $[0.05, 0.5]$\\
    Gauss2 & $\mathcal{N}(0.2, 1)$ & $[0.05, 1]$\\
    MixGauss1 & $\mathcal{N}(0.1, 1)$ $90\%$, $\mathcal{N}(0.5, 1)$ $10\%$  & $[0.05, 0.5]$\\
    MixGauss2 & $\mathcal{N}(0.2, 1)$ $90\%$, $\mathcal{N}(1, 1)$ $10\%$  & $[0.05, 1]$\\
    \bottomrule
  \end{tabular}
  \end{footnotesize}
  \caption{Distributions of Personalized LDP Budgets $\epsilon^l$. $\mathcal{U}$, $\mathcal{N}$ represents Uniform and Gaussian Distribution respectively. Clip range $[a,b]$ denotes any value $g$ outside the range $[a,b]$ is clipped by $\max(a, g)$ or $\min(b, g)$.}
  \label{table_eps_dist}
\end{table} 

\begin{table}
  \centering
  \begin{footnotesize}
  \begin{tabular}{lllll}
    \toprule
   Frameworks & $\qquad$ $\epsilon^{ul}$ & $\epsilon^c$ & $\epsilon^{uc}$ & Accuracy  \\
    \midrule
    Non-Private & $\quad \infty$ &  $\infty$ & $\infty$ &  \textbf{84.35}\% \\
    LDP-Min & \quad 392.5 & 0.05 & 40.1 & 56.11\% \\
    PLDP & 392.5 $\sim$ 7850 & 1  & 7850 & 77.54\%\\
    UniS   & 392.5 $\sim$ 7850 &  0.069\footnotemark[1] & 76.9 & 77.54\% \\
    APES & 392.5 $\sim$ 7850  &  0.057 & 57.6 & \textbf{79.67}\%\\
    S-APES & 78.5 $\sim$ 1570  &  0.057 & \textbf{25.6} & 78.14\%\\
    \bottomrule
  \end{tabular}
  \end{footnotesize}
  \caption{Privacy and Utility under Unifrom2 LDP. $\epsilon^{ul}$: local user level, $\epsilon^c$: central dimension level, $\epsilon^{uc}$: central user level privacy budgets.}
  \label{table_eps_amp}
\end{table} 
\footnotetext[1]{
	\textcolor{black}{Since the true central privacy under classic Laplace Mechanisms with varied $\epsilon^l_i$ is unbounded, $\epsilon^c$ of UniS in Tab. 3 is best considered as an approximation when the $\epsilon^l_i$ of different users are very similar to each other.}}

\paragraph{Privacy Amplification Effect}
In Fig. \ref{fig:priv_bounds}, we provide numerical evaluations for privacy amplification effect under fixed personalized LDP settings. 
Given dimension-level local privacy $\epsilon^l \in [0.05,1]$, we observe following results: 
(\romannumeral1) our bounds achieve the strongest central privacy with the smallest value of $\epsilon^c$ compared to baseline bounds under the same $n$. The bound gets sharper especially when $\epsilon^l$ concentrates on smaller values. E.g., for the same range that $\epsilon^l \in [0.05,1]$, most $\epsilon^l_i$ in Gauss2 are smaller than $\epsilon^l_i$ in Uniform2, which leads to lower $\epsilon^c$. This effect comes from the EoN analysis, by which privacy contribution of each local perturbation is taken into consideration. 
(\romannumeral2) The amplification effect gets stronger when more datapoints are shuffled, as more randomness is introduced for obfuscation. When $n$ grows, almost all the privacy bounds $\epsilon^c$ reduce.  
Moreover, Fig. 5 and 6 in Appendix C demonstrate EoN gives a more obvious amplification effect when the range of $\epsilon^l$ gets larger.

\paragraph{Stability of Clip-Laplace Mechanism}
Fig. \ref{fig:acc_norm} shows a relatively mild impact of clip bound $C$ on Clip-Laplace perturbation. We compare model accuracy by adopting Clip-Laplace Mechanism (CLap for short) and classic Laplace mechanism (Lap for short) in APES separately. Overall, the highest accuracy is obtained with CLap when $C=0.1$. CLap performs well especially for large $C$, while Lap is only good at small $C$. It implies CLap may be suitable for perturbing gradients with larger norms.
Fig. 7 in Appendix C explores why CLap adapts to varying parameters. The variance of CLap is more stable compared to Lap for the same level of LDP when $C$ changes. As a price of low variance resulting from the limited output range, a larger bias is introduced into perturbation (cf. Fig. 8 in Appendix C).

\paragraph{Performance of Post-Sparsification}
We evaluate the parameters and the effectiveness of post-sparsification technique (\emph{ps} for short) with Uniform2.
(\romannumeral1) The trade-off between accuracy and privacy of \emph{ps} is discussed above, while the knob is sparsification ratio $b/d$. In Fig. \ref{fig:topk}, the model with \emph{ps} achieves almost optimal accuracy as APES when $b/d=0.2$, \textcolor{black}{hence only smaller ratios are evaluated.} As $b/d$ grows, fewer dimensions are uploaded and the accuracy falls. In return, privacy cost is saved.
(\romannumeral2) Then we observe that \emph{ps} is more effective than random-sparification (\emph{rs} for short) which randomly select dimensions with $b/d$. Specifically, model accuracy based on \emph{rs} is lower than \emph{ps} for all the ratios, and dramatically drops when $b/d$ gets larger, as fewer informative dimensions are uploaded by \emph{rs}.


\begin{figure}[!t]
\centering
\subfigure[Privacy Bounds] 
	{\includegraphics[height=1.3in,width=1.7in,angle=0,trim=0 0 5 5,clip]{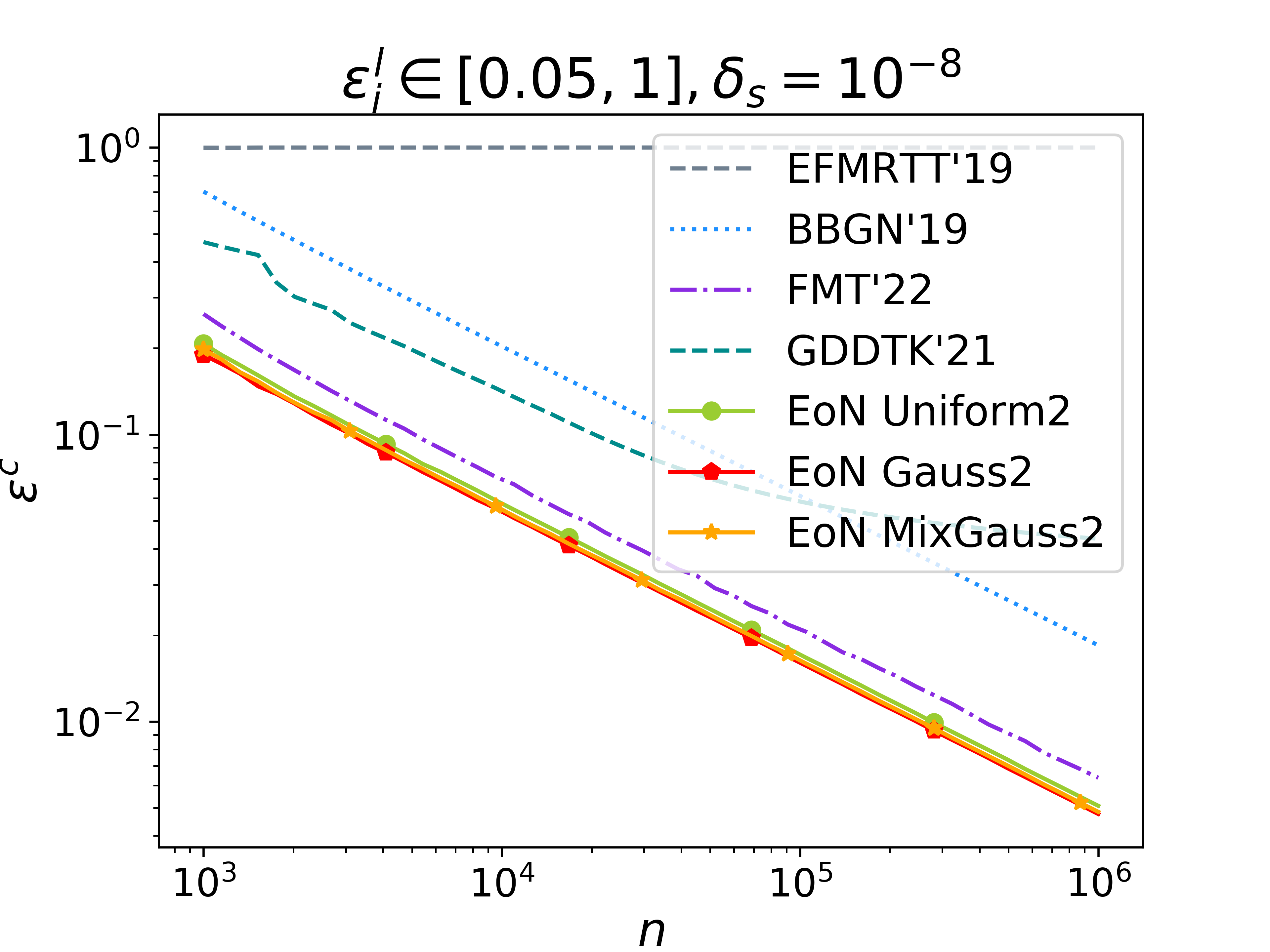}
	\label{fig:priv_bounds}}\hspace{-3mm}
\subfigure[Impact of $\epsilon^l$] {\includegraphics[height=1.3in,width=1.6in,angle=0,trim=0 0 5 5,clip]{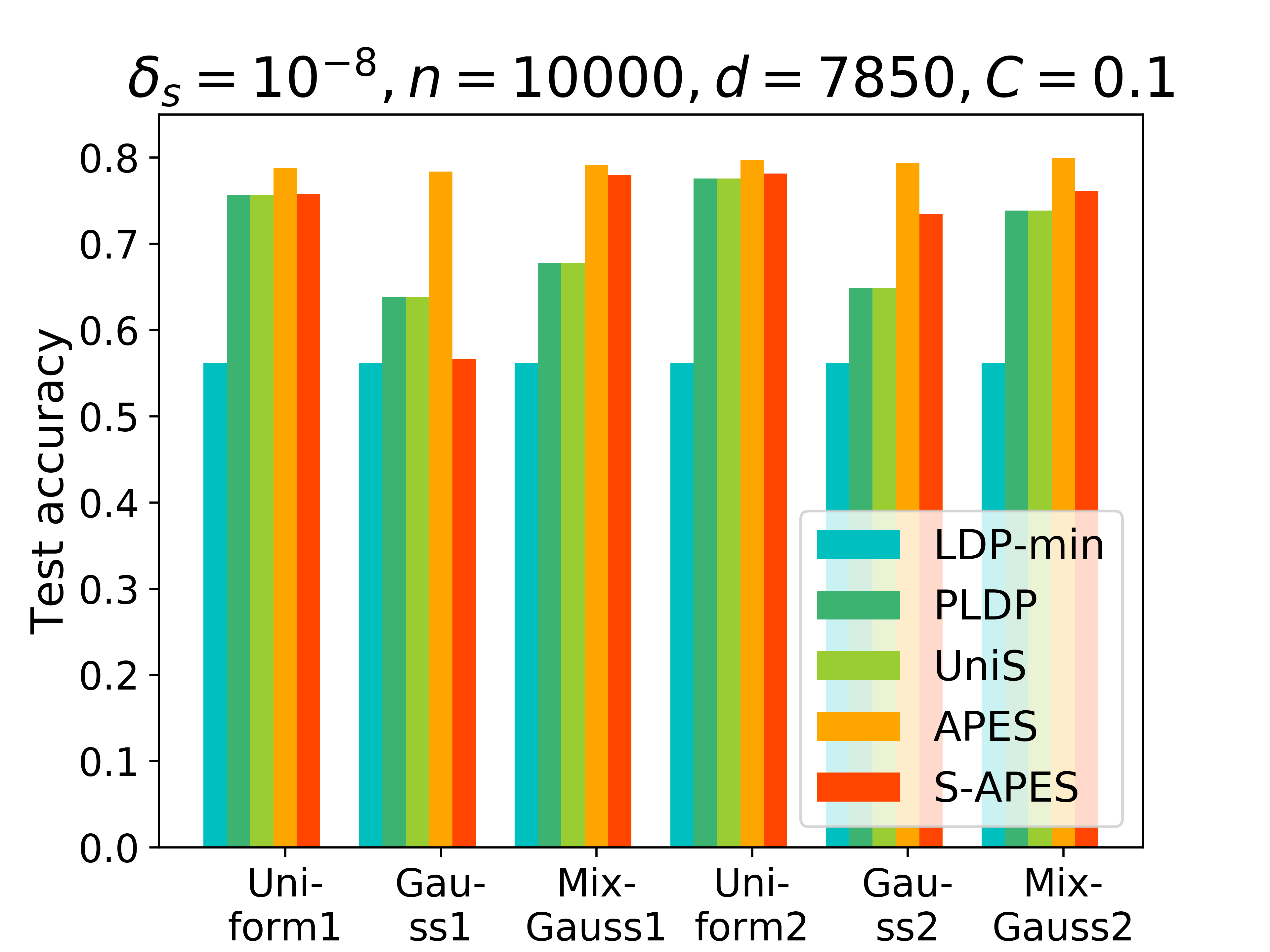}
	\label{fig:priv_utility_b}}\hspace{-3mm}
\caption{Privacy Bounds and Utility}
\end{figure}

\begin{figure}[!t]
\centering
\subfigure[Impact of $C$] 
	{\includegraphics[height=1.3in,width=1.7in,angle=0,trim=0 0 5 5,clip]{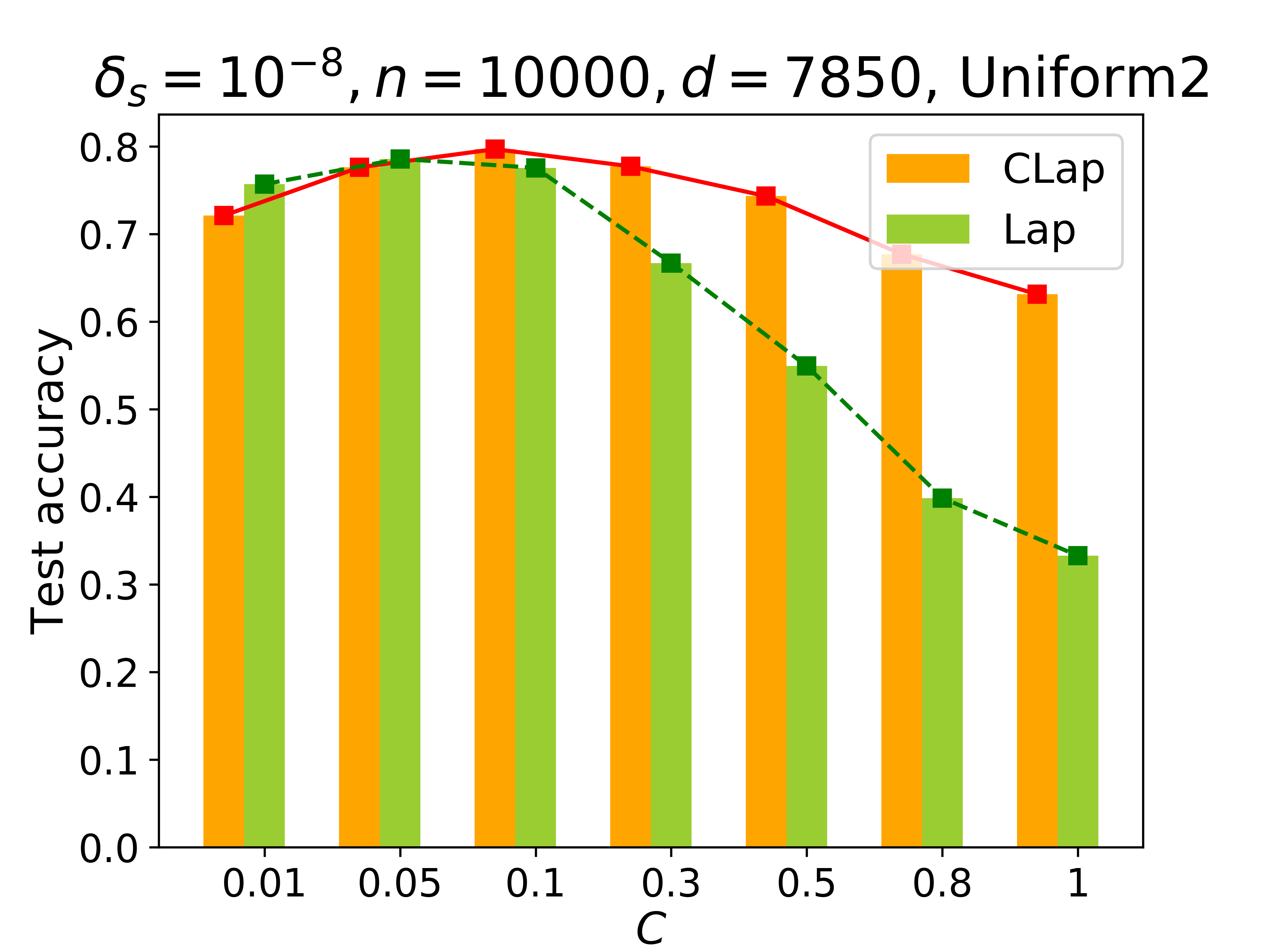}
	\label{fig:acc_norm}}\hspace{-3mm}
\subfigure[Impact of \emph{ps}] {\includegraphics[height=1.3in,width=1.6in,angle=0,trim=0 0 5 5,clip]{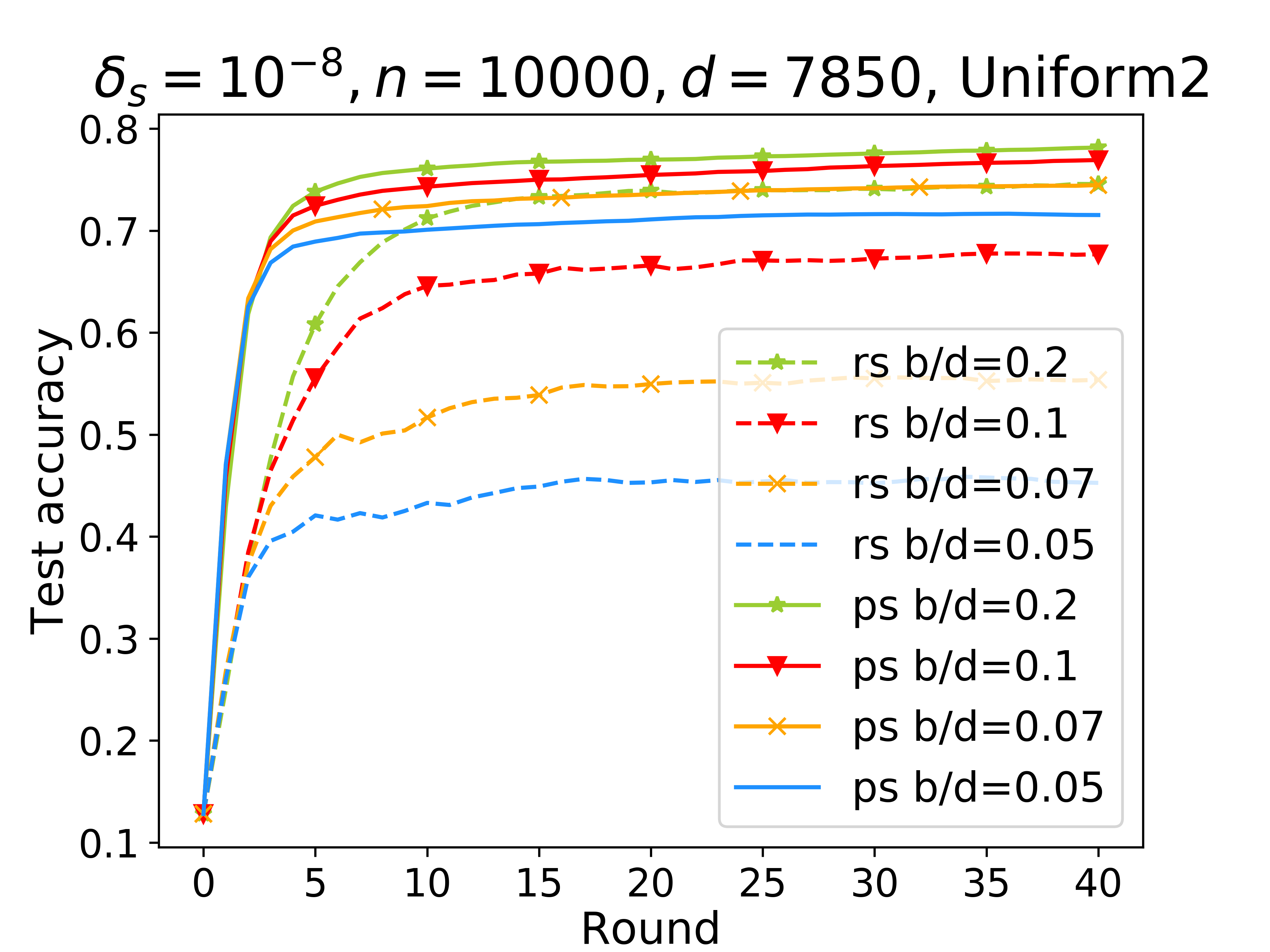}
	\label{fig:topk}}\hspace{-3mm}
\caption{Mechanisms and Strategies}
\end{figure}

\section{Conclusion}
This work focuses on personalized private federated learning. To balance privacy and utility, we propose privacy amplification frameworks with shuffle model under personalized LDP.  
Comprehensive evaluations on the public dataset confirm that our frameworks improve central privacy by reducing $\epsilon^{uc}$ up to 66.7\% compared to existing work with comparable or higher accuracy. 

\textcolor{black}{In the future, we intend to extend the work in several directions. First, we will explore the effectiveness of the work for larger models as sharing more parameters requires higher standards for both LDP performance and communication efficiency. Future improvements on sparsification techniques may alleviate the concern. 
Second, it may be possible to adapt the work to non-IID data distribution settings, hence more elaborate calibration for skewed gradients is required.}

\appendix

\section*{Acknowledgments}
We would like to thank all the anonymous reviewers for their time and efforts on our manuscript, their insightful comments and valuable suggestions help us shape the final draft. Our work is supported by National Natural Science Foundation of China (62072460, 62076245, 62172424,62276270), Beijing Natural Science Foundation (4212022), National Science Foundation (NSF) CNS-2124104, CNS-2125530, CNS-1952192, National Institute of Health (NIH) R01LM013712, R01ES033241, UL1TR002378, Cisco Research University Award \#2738379, Fundamental Research Funds for the Central Universities, and Research Funds of Renmin University of China (21XNH180).

\bibliography{aaai23}

\clearpage
\section*{Appendix}
This Appendix includes: (\romannumeral1) Section A: omitted proofs of privacy theorems. (\romannumeral2) Section B: omitted proofs of convergence analysis. (\romannumeral3) Section C: additional experiment results. 

\subsection*{A \quad Privacy Analysis}
\subsubsection*{Proof of Theorem 1}
\label{appendix:priv}
Let $x$ and $x' \in \mathcal{D}^d$ be such that $||x-x'|| \leq 1$, $f(x)$ and $f(x') \in [-\Delta f/2, \Delta f/2]$, $f(x')=f(x)+a$.
	\begin{footnotesize}
	\begin{align*}
		&\frac{\Pr[M(x)=z]}{\Pr[M(x')=z]} \nonumber \\
		&= \prod_{k=1}^d(\frac{1/S \cdot \exp(-|f(x)_{k}-z_{k}|/\lambda)}{1/S' \cdot \exp(-|f(x')_{k}-z_{k}|/\lambda)}) \nonumber \\
		&= \prod_{k=1}^d(\frac{1-\frac{1}{2}\exp(\frac{\epsilon(-A+f(x)_{k}+a_{k})}{\Delta f})-\frac{1}{2}\exp(\frac{\epsilon(-A-f(x)_{k}-a_{k})}{\Delta f})}{1-\frac{1}{2}\exp(\frac{\epsilon(-A+f(x)_{k})}{\Delta f})-\frac{1}{2}\exp(\frac{\epsilon(-A-f(x)_{k})}{\Delta f})} \nonumber  \\
		&\cdot \exp( \frac{\epsilon|f(x)_{k}-f(x')_{k}|}{\Delta f}))  \\ 
		& \leq \exp(\frac{\epsilon \cdot ||\Delta f/2 - (-\Delta f/2)||_1}{\Delta f}))   \\ 
		& \leq \exp(\epsilon)  \nonumber 
	\end{align*}
	\end{footnotesize}

\subsubsection*{Proof of Lemma 4}
Let $g_s$ and $g_t \in \mathcal{D}^d$ be such that $||g_s-g_t|| \leq 1$, $f(g_s)$ and $f(g_t) \in [-C, C]$ and $\Delta f = 2C$, we have
\begin{footnotesize}
\begin{align}
	& \frac{\Pr[M_i(g_s)=z]}{\Pr[M_j(g_t)=z]} 
	= \prod_{k=1}^d(\frac{1/(b_iS^{(0)}_i) \cdot \exp(-|f(g_s)_{k}-z_{k}|/b_i)}{1/(b_jS^{(1)}_j) \cdot \exp(-|f(g_t)_{k}-z_{k}|/b_j)}) \nonumber \\
	&\leq \prod_{k=1}^d(\frac{1-\frac{1}{2}\exp(\frac{\epsilon^l_j(-A+f(g_t)_{k})}{\Delta f})-\frac{1}{2}\exp(\frac{\epsilon^l_j(-A-f(g_t)_{k})}{\Delta f})}{1-\frac{1}{2}\exp(\frac{\epsilon^l_i(-A+f(g_s)_{k})}{\Delta f})-\frac{1}{2}\exp(\frac{\epsilon^l_i(-A-f(g_s)_{k})}{\Delta f})} \nonumber \\
	&\cdot \exp( \frac{|\epsilon^l_j f(g_t)_{k}-\epsilon^l_i f(g_s)_{k}| + |z_{k}(\epsilon^l_i - \epsilon^l_j) |}{\Delta f}) \cdot \frac{\epsilon^l_i}{\epsilon^l_j}) \label{eq_nd_clap}
\end{align}
\end{footnotesize}
When $\epsilon^l_j f(g_t)_{k} \geq \epsilon^l_i f(g_s)_{k}$, we define $l(u) = (1-\frac{1}{2}\exp(\frac{\epsilon{(-A+u)}}{\Delta f})-\frac{1}{2}\exp(\frac{\epsilon(-A-u)}{\Delta f})) \cdot \exp(\frac{\epsilon u}{\Delta f})$, the maximum and minimum of $l(u)$ is $l(C)$ and $l(-C)$ separately when $C \leq A$.
Then the right hand of Eq. (\ref{eq_nd_clap}) is bounded:
\begin{footnotesize}
\begin{align*}
	& \frac{\Pr[M_i(g_s)=z]}{\Pr[M_j(g_t)=z]} 
	 \leq \frac{1-\frac{1}{2}\exp(\frac{{\epsilon^l_j(-A+C)}}{2C})-\frac{1}{2}\exp(\frac{\epsilon^l_j(-A-C)}{2C})}{1-\frac{1}{2}\exp(\frac{{\epsilon^l_i(-A+C)}}{2C})-\frac{1}{2}\exp(\frac{\epsilon^l_i(-A-C)}{2C})} \\
	&\cdot \frac{\epsilon^l_i}{\epsilon^l_j} \cdot \exp(\frac{(\epsilon^l_i+\epsilon^l_j)}{2} + \frac{A|\epsilon^l_i-\epsilon^l_j|}{2C})
\end{align*}
\end{footnotesize}
In particular, when $A=C$, the bound is tight as follows:
\begin{footnotesize}
\begin{align}
	& \frac{\Pr[M_i(g_s)=z]}{\Pr[M_j(g_t)=z]}
	 \leq \exp(\frac{|\epsilon^l_i - \epsilon^l_j|}{2}) \cdot \frac{\exp(\epsilon^l_j/2)-\exp(-\epsilon^l_j/2)}{\exp(\epsilon^l_i/2)-\exp(-3\epsilon^l_i/2)} \nonumber \\
	& \cdot \frac{\epsilon^l_i}{\epsilon^l_j} \leq \frac{\epsilon^l_i}{\epsilon^l_j} \cdot \frac{1-e^{-\epsilon^l_j}}{1-e^{-\epsilon^l_i}} \cdot e^{\max(\epsilon^l_i, \epsilon^l_j)} \label{eq_nd_clap_righthand}
\end{align}
\end{footnotesize}
Similarly, we get the same bound as Eq. (\ref{eq_nd_clap_righthand}) when $\epsilon^l_j f(g_t)_{k} < \epsilon^l_i f(g_s)_{k}$. Thus, the proof is completed.
This Lemma allows us to transform echos in the following part.

\subsubsection*{Proof of Lemma 6}

From the proof in main body, we have the intermediate result as Eq. (5), which is the following equation:
\begin{footnotesize}
\begin{align}
	\frac{\Pr[P(D)=\mathbf{z}]}{\Pr[P(D')=\mathbf{z}]} 
	= \frac{|T_0|+1}{|T|-|T_0|}
\end{align}
\end{footnotesize}
Recall that $|T| \sim \sum_{i=2}^n\sum_{j=1}^n {\rm Bern}(p_{ij}/n)$ and $|T_0| \sim {\rm Bin}(1/2, |T|)$, 
by Chernoff bound and Hoeffding's inequality , $T$ and $T_0$ are concentrated to certain values. Specifically, when $\sum_{i=2}^n\sum_{j=1}^n \frac{p_{ij}}{n} \geq 3\ln(4/\delta_s)$, 
\begin{footnotesize}
\begin{align}
	&|T-\sum_{i=2}^n\sum_{j=1}^n \frac{p_{ij}}{n}| \leq (3\ln (4/\delta)\sum_{i=2}^n\sum_{j=1}^n \frac{p_{ij}}{n})^{\frac{1}{2}} \label{eq_chernoff1}\\
	&|T_0-T/2| \leq (T/2 \ln(4/\delta))^{\frac{1}{2}} \label{eq_chernoff2}
\end{align}
\end{footnotesize}

Based on Eq. \eqref{eq_chernoff1} and \eqref{eq_chernoff2}, if let $\sum_{i=2}^n\sum_{j=1}^n \frac{p_{ij}}{n} \geq 16\ln(4/\delta_s)$,  the following equation is established with the probability $(1-\delta_s)$.
\begin{footnotesize}
\begin{align}
&\frac{|T_0|+1}{|T|-|T_0|} 
\leq \frac{|T|/2 + (|T|/2\ln(4/\delta_s))^{\frac{1}{2}} + 1}{|T|/2 - (|T|/2\ln(4/\delta_s))^{\frac{1}{2}}} \nonumber \\
&\leq 1+\frac{8(\ln(4/\delta_s))^{\frac{1}{2}}}{(\sum\limits_{i=2}^n\sum\limits_{j=1}^n \frac{p_{ij}}{n})^{\frac{1}{2}}} + \frac{8}{\sum\limits_{i=2}^n\sum\limits_{j=1}^n \frac{p_{ij}}{n}} \label{eq_generalbound_e0} 
\end{align}
\end{footnotesize}
Hence the divergence between $P(D)$ and $P(D')$ is bounded without considering the randomness of user 1. For briefness, we set $1+\frac{8(\ln(4/\delta_s))^{\frac{1}{2}}}{(\sum\limits_{i=2}^n\sum\limits_{j=1}^n \frac{p_{ij}}{n})^{\frac{1}{2}}} + \frac{8}{\sum\limits_{i=2}^n\sum\limits_{j=1}^n \frac{p_{ij}}{n}} = e^{\epsilon_0}$.

Next, we take the randomness of the mechanism on $g_1$ and $g'_1$ with certain privacy budget $\epsilon^*$ into consideration, according to previous work \cite{kairouz2015composition,feldman2022hiding}. Formally, we assume $G$ (or $G'$) as the index of the mapping element from $\tilde{g}_1$ (or $\tilde{g}'_1$), and also apply degraded privacy (cf. Lemma \ref{degraded_priv}) on $\tilde{g}_1$ and $\tilde{g}'_1$. The mapping event on the $\tilde{g}_1$ (or $\tilde{g}'_1$) is defined as
 $U_G = \{ \rho^{(1)} \text{ for } p_1 \text{, } {\rho'}^{(1)} \text{ for } 1-p_1 \}$
 and
 $U'_G = \{ \rho^{(1)} \text{ for } p'_1 \text{, } {\rho'}^{(1)} \text{ for } 1-p'_1 \}$, where $p_1 = {e^{\epsilon^*}}/({1+e^{\epsilon^*}})$ and $p'_1={1}/({1+e^{\epsilon^*}})$.

Then we consider the mapping event $U_S$ from $(n-1)$ users of process $P$ except for $\tilde{g}_1$ (or $\tilde{g}'_1$). To reach the mixed output $\textbf{z}$ with the same number of $\rho^{(1)}$ or ${\rho'}^{(1)}$, $U_S$ needs to satisfy the following situation:
\begin{footnotesize}
\begin{align*}
& U_S = \{ U_T \text{ with } p_1 \text{, }  U'_T \text{ with } 1-p_1 \} \\
& U'_S = \{ U_T \text{ with } p'_1 \text{, } U'_T \text{ with } 1-p'_1 \text{, } \} \\
\end{align*}
\end{footnotesize}	
which can be written as:
\begin{footnotesize}
\begin{align*}
& U_S = \{ U_T\cup U'_T \text{ with } 1-p_1 \text{, }  U_T \text{ with } 2p_1-1 \} \\
& U'_S = \{ U_T\cup U'_T \text{ with } p'_1 \text{, } U'_T \text{ with } 1-2p'_1 \} \\
\end{align*}
\end{footnotesize}	
Hence we can bound the divergence between $U_S$ and $U'_S$:
\begin{footnotesize}
\begin{align}
&\frac{\Pr[P(D)=\textbf{z}]-\delta^c}{\Pr[P(D')=\textbf{z}]} = \frac{\Pr[U_G]\Pr[U_S|U_G]-\delta^c}{\Pr[U'_G]\Pr[U'_S|U'_G]} \nonumber \\
& = \frac{p_1(2(1-p_1)(\frac{1}{2})\Pr[U_T \cup U'_T] + (2p_1-1)\Pr[U_T])-\delta^c}{(1-p'_1)((2p'_1(\frac{1}{2})\Pr[U_T \cup U'_T] + (1-2p'_1)\Pr[U'_T]))} \label{eq_generalbound_Us}
\end{align}
\end{footnotesize}	
For convenience, we define the probabilities as $A_1 = (\frac{1}{2})\Pr[U_T \cup U'_T]$, $A_2 = \Pr[U_T]$, $A'_2=\Pr[U'_T]$, $2(1-p_1) = 2p'_1 = (1-p)$ and $(2p_1-1) = (1-2p'_1) = p$. Consider the bound by Eq. \eqref{eq_generalbound_e0} that $A_2 \leq e^{\epsilon_0}A'_2 + \delta_s$, and $A_2 \leq e^{\epsilon_0}A_1 + \delta_s$ according to hockey-stick divergence, we can rewrite Eq. \eqref{eq_generalbound_Us} as
\begin{footnotesize}
\begin{align*}
&\frac{\Pr[P(D)=\textbf{z}]-\delta^c}{\Pr[P(D')=\textbf{z}]}
= \frac{p_1((1-p)A_1 + pA_2)-p\delta_s}{(1-p'_1)((1-p)A_1 + pA'_2)} \\
& \leq \frac{p_1((1-p)A_1 + p(\min\{A_1, A'_2\} + e^{\epsilon_0}\min\{A_1, A'_2\}))-p\delta_s}{(1-p'_1)((1-p)A_1 + pA'_2)} \\
& \leq \frac{p_1((1-p)A_1 + p(A'_2 + e^{\epsilon_0}((1-p)A_1+ pA'_2)))-p\delta_s}{(1-p'_1)((1-p)A_1 + pA'_2)} \\
& \leq \frac{(1 + p(e^{\epsilon_0}-1))((1-p)A_1+ pA'_2))}{(1-p)A_1 + pA'_2} \\
& = 1 + p(e^{\epsilon_0}-1) \\
& = 1+ \frac{e^{\epsilon^*}-1}{e^{\epsilon^*}+1}(\frac{8(\ln(4/\delta_s))^{1/2}}{(\sum\limits_{i=2}^n\sum\limits_{j=1}^n \frac{p_{ij}}{n})^{1/2}} + \frac{8}{\sum\limits_{i=2}^n\sum\limits_{j=1}^n \frac{p_{ij}}{n}})
\end{align*}
\end{footnotesize}

As $\epsilon^c \leq \ln(\frac{\Pr[P(D)=\textbf{z}]-\delta^c}{\Pr[P(D')=\textbf{z}]})$, the general bound of process $P$ in APES is proved, and $\delta^c \leq \frac{e^{\epsilon^*}-1}{e^{\epsilon^*}+1} \delta_s$.

\subsubsection{Proof of Theorem 4} For the worst case, user 1 adpots $\max(\epsilon^l_j)$ as her privacy budget $\epsilon^*$, which leads to an upper bound $\epsilon^c$ derived from Lemma \ref{lemma:general_privacy_bound}. Therefore, $P(D)$ and $P(D')$ are $(\epsilon^c, \delta^c)$-DP:
	\begin{footnotesize}
	\begin{equation*}
		\epsilon^c \leq  \ln(1+ \frac{e^{\max(\epsilon^l_j)}-1}{e^{\max(\epsilon^l_j)}+1}(\frac{8(\ln(4/\delta_s))^{1/2}}{(\sum\limits_{i=2}^n\sum\limits_{j=1}^n \frac{p_{ij}}{n})^{1/2}} + \frac{8}{\sum\limits_{i=2}^n\sum\limits_{j=1}^n \frac{p_{ij}}{n}}))
	\end{equation*}
	\end{footnotesize}
	when $\sum_{i=2}^n\sum_{j=1}^n \frac{p_{ij}}{n} \geq 16\ln(4/\delta_s)$, $\delta_s \in [0,1]$ and $\delta^c \leq \frac{e^{\max(\epsilon^l_j)}-1}{e^{\max(\epsilon^l_j)}+1}\delta_s$.
Given Lemma \ref{lemma_nd_of_clap}, $p_{ij}$ equals
	$(\frac{\epsilon^l_i}{\epsilon^l_j} \cdot \frac{1-e^{-\epsilon^l_j}}{1-e^{-\epsilon^l_i}} \cdot e^{-\max(\epsilon^l_i, \epsilon^l_j)})$.

\subsection*{B \quad Convergence Analysis}
To give the final convergence upper bound of the frameworks after $T$ aggregations, we make the following analysis: (\romannumeral1) Show the assumptions on loss function. (\romannumeral2) Derive a general form for APES with general perturbation. (\romannumeral3) Consider the perturbation of Clip-Lalpalce Mechanism in bound. (\romannumeral4) Bring the effect of calibration into the bound.

\subsubsection*{Step (\romannumeral1) Assumptions}
To analyze the convergence of APES, we first make following assumptions of loss function \cite{li2020federated, wei2020federated}.
\begin{enumerate}
    \item The function $\{F_i\}$ of user $i$ for $i \in [n]$ is non-convex, $L$-Lipschitz smooth, and there exists $L > 0$, such that $\nabla^2F_i \succeq L \mathcal{I}$ with $\bar{\mu} := L + \mu > 0$.
    \item The functions $\{F_i\}$ satisfy $\beta$-Polyak-Lojasiewicz (PL) condition, which implies $f(\tilde{w})-f(w^*) \le \frac{1}{2\beta}\Vert \nabla f(\tilde{w}) \Vert^2$, and $w^*$ is the optimal parameters for loss function.
    \item $f(\cdot)$ satisfies $l$-Lipschitz continuous condition.
    \item $w^*$ is the optimal solution for minimum objective function: $w^* = \min_w h(w;w_t)$. We assume $h(w;w_t)$ is $\alpha$-close to minimum function: $\Vert \nabla h(w^*;w_t) \Vert \le \alpha \Vert \nabla h(w;w_t) \Vert$ where $h(w;w_t) = F_i(w)+\frac{\mu}{2}\Vert w-w_t\Vert^2$.
\end{enumerate}
We also measure the dissimilarity between users in Definition (\ref{lemma_dissimi}), which is also assumed in previous works that we refers above.

\begin{Defs}[User Dissimilarity]
\label{lemma_dissimi}
The local loss functions $F_i$ of users are B-locally dissimilar at $w$ if $\mathbb{E}_i[{\Vert \nabla F_i(w)\Vert ^2}] \le \Vert \nabla f(w) \Vert^2B^2$. Here $B(w)=\sqrt{\frac{\mathbb{E}_i[{\Vert \nabla F_i(w)\Vert^2}]}{\Vert \nabla f(w) \Vert^2}}$ for $\Vert \nabla f(w) \Vert \ne 0$.
\end{Defs}

\subsubsection*{Step (\romannumeral2) Proof of General Form} We give the general for of convergence upper bound by yielding the bound in one aggregation and inducting to T aggregations.

\begin{lemma}[Bound of Single Aggregation]
\label{lemma_conv_single_iter}
For a global function $f(w)=\mathbb{E}[F_i(w)]$ in analyzer, the expectation of difference of $f(w)$ for a single aggregation between $t$-th and $(t+1)$-th round is bounded as follows:
\begin{footnotesize}
\begin{align*}
&\mathbb{E}[f(\tilde{w}^{(t+1)})-f(\tilde{w}^{(t)})] \\
&\le 
(\frac{\alpha B-1}{\mu} + \frac{LB(\alpha+1)}{\mu\bar{\mu}} 
+ \frac{LB^2(1+\alpha)^2}{2\bar{\mu}^2})\mathbb{E}[\Vert \nabla f(\tilde{w}^{(t)}) \Vert^2] \\
&+ (\frac{1}{\mu} + \frac{BL(1+\alpha)}{\bar{\mu}}) \mathbb{E}[\Vert \nabla f(\tilde{w}^{(t)}) \Vert \Vert \eta^{(t+1)} \Vert] 
+ \frac{L}{2} \mathbb{E}[\Vert \eta^{(t+1)} \Vert^2]
\end{align*}
\end{footnotesize}
where $\eta$ denotes the perturbation term introduced by DP mechanisms, and $\eta^{(t)}=\frac{1}{n}\sum_i(\tilde{w_i}^{(t)}-w_i^{(t)})$.
\end{lemma}

\begin{proof}
According to the assumption 1, $F_i$ is $L$-Lipschitz smooth,
\begin{footnotesize}
\begin{align}
	\label{eq7}
    &f(\tilde{w}^{(t+1)}) 
    \le f(\tilde{w}^{(t)}) + <\nabla f(\tilde{w}^{(t)}), \tilde{w}^{(t+1)}-\tilde{w}^{(t)}> \nonumber \\
    &+ \frac{L}{2}\Vert \tilde{w}^{(t+1)}- \tilde{w}^{(t)}\Vert^2
\end{align}
\end{footnotesize}
where $\tilde{w}^{(t+1)}= \mathbb{E}[F_i(\tilde{w}^{(t+1)})]$,$\tilde{w}^{(t)}=\frac{1}{\eta}\sum^n_{i=1}{w_i^{(t+1)}} + \eta^{(t+1)}$ and $\eta^{(t)} = \sum^n_{i=1}{\eta_i^{(t)}}$.
To bound $f(\tilde{w}^{(t+1)})-f(\tilde{w}^{(t)})$, we have to bound $\Vert \tilde{w}^{(t+1)}- \tilde{w}^{(t)}\Vert$ and $(\tilde{w}^{(t+1)}- \tilde{w}^{(t)})$ separately. 

First, we focus on $\Vert \tilde{w}^{(t+1)}- \tilde{w}^{(t)}\Vert$.

From assumption 3 with $\alpha$-clossness we have 
\begin{equation}
    \label{eq11}
    \Vert \nabla h_k(w^{(t+1)}; \tilde{w}^{(t)})\Vert \le \alpha \Vert \nabla F_k(\tilde{w}^{(t)})\Vert
\end{equation}
Since $\nabla^2 h_k(w;\tilde{w}^{(t)})=\bar{\mu}>0$, $h_k$ is $\bar{\mu}$-strong convex. Let $w^*_k=\mathop{\arg\min}_w h_k(w;\tilde{w}^{(t)})$, we have 
\begin{footnotesize}
\begin{align}
    \label{eq8}
    &\bar{\mu} \Vert {w^*_i}^{(t+1)}-w^{(t+1)}_i \Vert \nonumber \\
    &\le \Vert \nabla h_i( {w^*}^{(t+1)};\tilde{w}^{(t)}) - \nabla h_i( w^{(t+1)};\tilde{w}^{(t)})\Vert \nonumber \\
    &= \Vert 0-\nabla h_i( w^{(t+1)};\tilde{w}^{(t)})\Vert = \alpha \Vert \nabla F_i(\tilde{w}^{(t)})\Vert
\end{align}
\end{footnotesize}
From the strong convexity of $h_k$, we also know that 
\begin{footnotesize}
\begin{align}
    \label{eq9}
    &\bar{\mu} \Vert {w^*}^{(t+1)}_i-\tilde{w}^{(t)}_i \Vert \nonumber \\
    &\le \Vert \nabla h_i( w^{(t+1)^*};\tilde{w}^{(t)}) - \nabla h_i( \tilde{w}^{(t)};\tilde{w}^{(t)})\Vert \nonumber \\
    &= \Vert 0-\nabla h_i( \tilde{w}^{(t)};\tilde{w}^{(t)})\Vert \nonumber \\
    &= \Vert \nabla F_i(\tilde{w}^{(t)}) + \mu(\tilde{w}^{(t)}-\tilde{w}^{(t)})\Vert \nonumber \\
    &= \Vert \nabla F_i(\tilde{w}^{(t)}) \Vert
\end{align}
\end{footnotesize}
Based on triangle inequality, B-user dissimilarity and Eq. \eqref{eq7}\eqref{eq8}\eqref{eq9}, $\Vert \tilde{w}^{(t+1)}- \tilde{w}^{(t)}\Vert$ is bounded:
\begin{footnotesize}
\begin{align}
    \label{eq12}
    &\Vert \tilde{w}^{(t+1)}- \tilde{w}^{(t)}\Vert  \nonumber \\
    &\le \Vert \tilde{w}^{(t+1)} - \tilde{w^*}^{(t+1)}\Vert + \Vert \tilde{w^*}^{(t+1)} - w^{(t)}\Vert + \Vert \eta^{(t+1)} \Vert \nonumber \\
    &\le \mathbb{E}_i[\Vert w^{(t+1)}_i-{w^*}^{(t+1)}_i \Vert + \Vert {w^*}^{(t+1)}_i-w^{(t)}_i \Vert ]  + \Vert \eta^{(t+1)} \Vert\nonumber \\
    &\le \frac{1+\alpha}{\bar{\mu}} \mathbb{E}_i[\Vert \nabla F_i(\tilde{w}^{(t)}) \Vert]  + \Vert \eta^{(t+1)} \Vert
\end{align}
\end{footnotesize}

Then, let us bound $(\tilde{w}^{(t+1)}- \tilde{w}^{(t)})$.
Differentiate $h_i$, we obtain
\begin{footnotesize}
\begin{align}
    &\tilde{w}^{(t+1)}- \tilde{w}^{(t)}
    = w^{(t+1)} - \tilde{w}^{(t)} + \eta^{(t+1)} \nonumber \\
    &\le \frac{1}{\mu}(\mathbb{E}[\nabla h_i(w^{(t+1)}; \tilde{w}^{(t)}) - \nabla F_i(w^{(t+1)})]) + \eta^{(t+1)} \nonumber \\
    &= \frac{1}{\mu}(\mathbb{E}[\nabla h_i(w^{(t+1)}; \tilde{w}^{(t)}) - \nabla F_i(w^{(t+1)}) + \nabla F_i(\tilde{w}^{(t)} ) \nonumber \\
    &- \nabla F_i(\tilde{w}^{(t)}))] + \eta^{(t+1)} \label{eq_w-w}
\end{align}
\end{footnotesize}
By triangle inequality and Eq. \eqref{eq11}\eqref{eq12}, Lemma \ref{lemma_dissimi} and the L-Lipschitz condition on $F_i$, We can bound the part $(\nabla h_i(w^{(t+1)}; \tilde{w}^{(t)}) - \nabla F_i(w^{(t+1)}) + \nabla F_i(\tilde{w}^{(t)} ))$ of Eq.\eqref{eq_w-w}.
\begin{footnotesize}
\begin{align}
    \label{eq10}
	&\Vert (\mathbb{E}[\nabla h_i(w^{(t+1)}; \tilde{w}^{(t)}) - \nabla F_i(w^{(t+1)}) + \nabla F_i(\tilde{w}^{(t)} )] \Vert  \nonumber \\
	& \leq \mathbb{E}[\Vert \nabla h_i(w^{(t+1)}; \tilde{w}^{(t)}) \Vert + \Vert \nabla F_i(w^{(t+1)}) - \nabla F_i(\tilde{w}^{(t)} ) \Vert] \nonumber \\
	&= (\alpha B + \frac{LB(1+\alpha)}{\bar{\mu}}) \mathbb{E}[\Vert \nabla F_i(w^{(t+1)}) \vert] \nonumber \\
    &= (\alpha B + \frac{LB(1+\alpha)}{\bar{\mu}}) \Vert f(\tilde{w}^{(t)})\Vert
\end{align}
\end{footnotesize}

At last, substitute Eq. \eqref{eq12}, \eqref{eq_w-w}, and \eqref{eq10} into \eqref{eq7}, Lemma \ref{lemma_conv_single_iter} is proved.
\begin{footnotesize}
\begin{align}
    \label{eq13}
    &\mathbb{E}[f(\tilde{w}^{(t+1)})- f(\tilde{w}^{(t)})] \nonumber \\
    &\le (\frac{\alpha B-1}{\mu} + \frac{LB(\alpha+1)}{\mu\bar{\mu}} + \frac{LB^2(1+\alpha)^2}{2\bar{\mu}^2})\mathbb{E}[\Vert \nabla f(\tilde{w}^{(t)}) \Vert^2] \nonumber \\
    &+ (\frac{1}{\mu} + \frac{BL(1+\alpha)}{\bar{\mu}}) \mathbb{E}[\Vert \nabla f(\tilde{w}^{(t)}) \Vert \Vert \eta^{(t+1)} \Vert] \nonumber \\
    &+ \frac{L}{2} \mathbb{E}[\Vert \eta^{(t+1)} \Vert^2]
\end{align}	
\end{footnotesize} 
The proof is completed. \qedsymbol

\end{proof}

%
By lemma \ref{lemma_conv_single_iter}, we can derive the convergence upper bound for T aggregations.

\begin{lemma}[General Form of Convergence Upper Bound]
\label{appendix:conv}
\label{th_converge_1}
The expected decrease in the global loss function $f(w) = \frac{1}{n}\sum_iF_i(w)$ after T aggregations is bounded as follows:
\begin{footnotesize}
\begin{align}
\label{eq_conv_PerS}
    &\mathbb{E}[f(\tilde{w}^{(T)})-f(w^*)]
    \le a_1^T\mathbb{E}[f(\tilde{w}^{(0)})-f(w^*)] \nonumber \\
    &+ \frac{a_1^T-1}{a_1-1}(a_2 \mathbb{E}[\Vert \eta \Vert] + a_3 \mathbb{E}[\Vert \eta \Vert^2])
\end{align}
\end{footnotesize}
where $a_1=1+\frac{2\beta(\alpha B-1)}{\mu} + \frac{2\beta LB(\alpha+1)}{\mu\bar{\mu}} + \frac{2\beta LB^2(1+\alpha)^2}{\bar{\mu}^2}, a_2=l(\frac{1}{\mu} + \frac{BL(1+\alpha)}{\bar{\mu}}), a_3=\frac{L}{2}$.
\end{lemma} 

\begin{proof}
To bound $\mathbb{E}[f(\tilde{w}^{(T)})-f(w^*)]$, we first transform it as follows:
\begin{footnotesize}
\begin{align*}
    &\mathbb{E}[f(\tilde{w}^{(T)})-f(w^*)+f(w^*)-f(\tilde{w}^{(t)})] \nonumber \\
    &\le (\frac{\alpha B-1}{\mu} + \frac{LB(\alpha+1)}{\mu\bar{\mu}} + \frac{LB^2(1+\alpha)^2}{2\bar{\mu}^2})\mathbb{E}[\Vert \nabla f(\tilde{w}^{(t)}) \Vert^2] \\
    &+ (\frac{1}{\mu} + \frac{BL(1+\alpha)}{\bar{\mu}}) \mathbb{E}[\Vert \nabla f(\tilde{w}^{(t)}) \Vert \Vert \eta^{(t+1)} \Vert] 
    + \frac{L}{2} \mathbb{E}[\Vert \eta^{(t+1)} \Vert^2]
\end{align*}
\end{footnotesize}
From assumption 3 and 4, we have 
$\mathbb{E}[f(\tilde{w}^{(t)})-f(w^*)] \le \frac{1}{2\beta} \Vert \nabla f(\tilde{w}^{(t)}) \Vert^2$
and
$\Vert \nabla f(\cdot) \Vert \le l $. Substract $f(w^*)$ from Eq.\eqref{eq13} on both sides, we have

\begin{footnotesize}
\begin{align}
    \label{eq_general_1}
    &\mathbb{E}[f(\tilde{w}^{(t)})-f(w^*)]  \nonumber \\
    &\le (1+\frac{2\beta(\alpha B-1)}{\mu} + \frac{2\beta LB(\alpha+1)}{\mu\bar{\mu}} + \frac{2\beta LB^2(1+\alpha)^2}{\bar{\mu}^2}) \nonumber \\
    &\cdot \mathbb{E}[f(\tilde{w}^{(t)})-f(w^*)] 
    + l (\frac{1}{\mu} + \frac{BL(1+\alpha)}{\bar{\mu}}) \mathbb{E}[\Vert \eta^{(t+1)} \Vert] \nonumber \\
    &+ \frac{L}{2} \mathbb{E}[\Vert \eta^{(t+1)} \Vert^2]
\end{align}
\end{footnotesize}
Considering expectations of perturbation term $\eta$ is the same in each epoch,  we define $\mathbb{E}[\Vert \eta^{(t)}\Vert] = \mathbb{E}[\Vert \eta \Vert]$ for $t \in [0, T]$. By Eq. \eqref{eq_general_1} we have
\begin{footnotesize}
\begin{align}
    \label{eq2}
    &\mathbb{E}[f(\tilde{w}^{(t)})-f(w^*)] 
    \le a_1\mathbb{E}[f(\tilde{w}^{(t)})-f(w^*)] 
    + a_2 \mathbb{E}[\Vert \eta^{(t+1)} \Vert] \nonumber \\
    &+ a_3 \mathbb{E}[\Vert \eta^{(t+1)} \Vert^2]
\end{align}
\end{footnotesize}
where $a_1=1+\frac{2\beta(\alpha B-1)}{\mu} + \frac{2\beta LB(\alpha+1)}{\mu\bar{\mu}} + \frac{2\beta LB^2(1+\alpha)^2}{\bar{\mu}^2}, a_2=l(\frac{1}{\mu} + \frac{Bh(\alpha+1)}{\bar{\mu}}), a_3=\frac{L}{2}$.

The bound of T aggregations is inducted with Eq.\eqref{eq2}.
\begin{footnotesize}
\begin{align*}
    &\mathbb{E}[f(\tilde{w}^{(T)})-f(w^*)] \nonumber \\
    &\le a_1(...(a_1(a_1(\mathbb{E}[f(\tilde{w}^{(0)})-f(w^*)]))))\nonumber \\
    &+ (a_1^0+a_1^1+...+a_1^{T-1})(a_2 \mathbb{E}[\Vert \eta \Vert]+ a_3 \mathbb{E}[\Vert \eta \Vert^2])\nonumber \\
    &= a_1^T\mathbb{E}[f(\tilde{w}^{(0)})-f(w^*)] 
    + \sum_{t=0}^{T-1}a_1^t(a_2 \mathbb{E}[\Vert \eta \Vert] + a_3 \mathbb{E}[\Vert \eta \Vert^2]) \nonumber\\
    &= a_1^T\mathbb{E}[f(\tilde{w}^{(0)})-f(w^*)] + \frac{a_1^T-1}{a_1-1}(a_2 \mathbb{E}[\Vert \eta \Vert] + a_3 \mathbb{E}[\Vert \eta \Vert^2])
\end{align*}
\end{footnotesize}
Theorem \ref{th_converge_1} is proved.
\end{proof}

\subsubsection{Step (\romannumeral3) Proof after CLap perturbation} The specific perturbation bias and variance introduced by Clip-Laplace Mechanism is considered in this section. We have $\tilde{g_i} \sim CLap(g_i, \lambda_i, C)$ where $\lambda_i=\Delta f / \epsilon_i$, and $g_i = \nabla f(x_i; w)$. Thus, local perturbed term $\eta_i$ equals $\tilde{g_i}-g_i$, and aggregated perturbed term $\eta$ is $\frac{1}{n}\sum_i \eta_i$.
\begin{lemma}[Bound with Clip-Laplacian perturbation]
By personalized privacy level $\epsilon^l=(\epsilon^l_1, ..., \epsilon^l_n)$, the convergence bound of APES with Clip-laplacian noises is given as:
\begin{align*}
   & \mathbb{E}[f(\tilde{w}^{(T)})-f(w^*)]
    \le a_1^T\Delta \\
    &+ \frac{a_1^T-1}{a_1-1}(O(a_2C/\min(\epsilon^l_i)) + O(a_3 C^2/\min(\epsilon^l_i)^2))
\end{align*}
where $\mathbb{E}[f(\tilde{w}^{(0)})-f(w^*)]$ is denoted by $\Delta$.
\end{lemma}

\begin{proof}
Consider local perturbation for each user $i$,
\begin{footnotesize}
\begin{align*}
    &\mathbb{E}[\eta_i^2]
    =\frac{1}{2S}((\lambda_i^2-(g_i+C+\lambda_i)^2)e^{\frac{-C-g_i}{\lambda_i}} \nonumber \\
    &+ (\lambda_i^2-(-g_i+C+\lambda_i)^2)e^{\frac{-C+g_i}{\lambda_i}}) + 2\lambda_i^2 \\
    & \leq \frac{(-4C^2-4\lambda_iC)e^{\frac{-2C}{\lambda_i}}}{1-e^{\frac{-2C}{\lambda_i}}} + 2\lambda_i^2
\end{align*}
\end{footnotesize}
We can obtain the bound of $\mathbb{E}[\Vert \eta \Vert^2]$.
\begin{footnotesize}
\begin{align}
    \label{eq_En2}
    &\mathbb{E}[\Vert \eta \Vert^2]
    = \mathbb{E}[\frac{1}{n}\sum_{i=1}^{n}\eta_i^2]
     = \frac{1}{n}\sum_{i=1}^{n}\mathbb{E}[\eta_i^2]
     \le \mathbb{E}[\eta_i^2]_{max} \nonumber \\
    & = \frac{-4C^2-8C/{\min(\epsilon^l_i})}{e^{\min(\epsilon^l_i)}-1} + \frac{8C^2}{\min(\epsilon^l_i)^2}
\end{align}
\end{footnotesize}
Then $\mathbb{E}[\Vert \eta \Vert]$ is also bounded as:
\begin{footnotesize}
\begin{align}
    \label{eq_En}
    &\mathbb{E}[\Vert \eta \Vert]) 
    = (\mathbb{E}[\Vert \eta \Vert^2]) - \mathbb{D}[\Vert \eta \Vert]))^{\frac{1}{2}}
    \le (\mathbb{E}[\Vert \eta \Vert^2]))^{\frac{1}{2}} \nonumber \\
    &= (\frac{-4C^2-8C/{\min(\epsilon^l_i})}{e^{\min(\epsilon^l_i)}-1} + \frac{8C^2}{\min(\epsilon^l_i)^2})^{\frac{1}{2}}
\end{align}
\end{footnotesize}
Substituting equation (\ref{eq_En2}) and (\ref{eq_En}) into (\ref{eq_conv_PerS}), the convergence bound is provided:
\begin{footnotesize}
\begin{align*}
    &\mathbb{E}[f(\tilde{w}^{(T)})-f(w^*)]
     \le a_1^T\mathbb{E}[f(\tilde{w}^{(0)})-f(w^*)] \\
    & + \frac{a_1^T-1}{a_1-1}(a_2(\frac{-4C^2-8C/{\min(\epsilon^l_i})}{e^{\min(\epsilon^l_i)}-1} + \frac{2C^2}{\min(\epsilon^l_i)^2})^{\frac{1}{2}}) \\ 
    &+ a_3(\frac{-4C^2-8C/{\min(\epsilon^l_i})}{e^{\min(\epsilon^l_i)}-1} + \frac{8C^2}{\min(\epsilon^l_i)^2})))
\end{align*}
\end{footnotesize}
\end{proof}
The proof is completed.

\subsubsection{Step (\romannumeral4) Proof after Calibration} Besides local perturbation, analyzer will also calibrate the noisy gradients centrally, which increases the convergence rate empirically. We provide the theoretical bound in this section.
\begin{lemma}[Bound after Calibration]
\label{lemma:bound_after_cali}
After Clip Laplacian perturbation with $\epsilon^l$ and Calibration in server, the convergence rate of APES is bounded as:
\begin{align*}
    & \mathbb{E}[f(\tilde{w}^{(T)})-f(w^*)]
    \le a_1^T\Delta \\
    & + \frac{a_1^T-1}{a_1-1}(O(a_2C/\min(\epsilon^l_i)) + O(a_3 C^2/\min(\epsilon^l_i)^2))
\end{align*}

\end{lemma}
\begin{proof}
In calibration, we estimate $\bar{g}$ by approximating $\mathbb{E}[\tilde{g}]$ with $\mathbb{E}[\tilde{\bar{g}}_i]$ for $i \in [n]$. The bias introduced by $\tilde{g}$ is removed with estimated $\bar{g}$ mostly, we denote the residual bias for each user by $\Delta \eta_i$, which equals to $\eta_i - \eta(\bar{g})$. Note that $\eta(\bar{g}) \in [-C, C]$, the norm of $\Delta \eta_i$ is bounded:
\begin{footnotesize}
\begin{align*}
    &\mathbb{E}[\Delta \eta_i^2]
    = \mathbb{E}[\eta_i^2] + \mathbb{E}[\eta_i^2] - 2\mathbb{E}[\eta_i] \mathbb{E}[\eta_i] 
    \leq \frac{(-8C^2-8\lambda_iC)e^{\frac{-2C}{\lambda_i}}}{1-e^{\frac{-2C}{\lambda_i}}} \\
    &+ \frac{(-8C^2+8\lambda_iC)e^{\frac{-4C}{\lambda_i}} - 8\lambda_iCe^{\frac{-2C}{\lambda_i}}}{(1-e^{\frac{-2C}{\lambda_i}})^2} + 6\lambda_i^2
\end{align*}
\end{footnotesize}

Therefore we obtain the bound of $\mathbb{E}[\Vert \Delta \eta_i \Vert^2]$.
\begin{footnotesize}
\begin{align}
    \label{eq_En2_cali}
    &\mathbb{E}[\Vert \Delta \eta_i \Vert^2]
    = \mathbb{E}[\frac{1}{n}\sum_{i=1}^{n}\Delta \eta_i^2]
     = \frac{1}{n}\sum_{i=1}^{n}\mathbb{E}[\Delta \eta_i^2]
     \le \mathbb{E}[\Delta \eta_i^2]_{max} \nonumber \\
    & = -\frac{8C^2}{e^{\min(\epsilon^l_i)}-1} + \frac{8C^2}{(e^{\min(\epsilon^l_i)}-1)^2} + \frac{24C^2}{\min(\epsilon^l_i)^2}
\end{align}
\end{footnotesize}
Then $\mathbb{E}[\Vert \Delta \eta_i \Vert]$ is also bounded as:
\begin{footnotesize}
\begin{align}
    \label{eq_En_cali}
    &\mathbb{E}[\Vert \Delta \eta_i \Vert]) 
    = (\mathbb{E}[\Vert \eta \Vert^2]) - \mathbb{D}[\Vert \eta \Vert]))^{\frac{1}{2}}
    \le (\mathbb{E}[\Vert \eta \Vert^2]))^{\frac{1}{2}} \nonumber \\
    &= (-\frac{8C^2}{e^{\min(\epsilon^l_i)}-1} + \frac{8C^2}{(e^{\min(\epsilon^l_i)}-1)^2} + \frac{24C^2}{\min(\epsilon^l_i)^2})^{\frac{1}{2}}
\end{align}
\end{footnotesize}
Substituting Eq. \eqref{eq_En2_cali} and \eqref{eq_En_cali} into (\ref{eq_conv_PerS}), the bound is provided:
\begin{footnotesize}
\begin{align*}
    &\mathbb{E}[f(\tilde{w}^{(T)})-f(w^*)]
     \le a_1^T\mathbb{E}[f(\tilde{w}^{(0)})-f(w^*)] \\
     & + \frac{a_1^T-1}{a_1-1}(a_2(-\frac{8C^2}{e^{\min(\epsilon^l_i)}-1} + \frac{8C^2}{(e^{\min(\epsilon^l_i)}-1)^2} + \frac{24C^2}{\min(\epsilon^l_i)^2})^{\frac{1}{2}}) \\ 
    &+ a_3(-\frac{8C^2}{e^{\min(\epsilon^l_i)}-1} + \frac{8C^2}{(e^{\min(\epsilon^l_i)}-1)^2} + \frac{24C^2}{\min(\epsilon^l_i)^2})
\end{align*}
\end{footnotesize}
Ideally, when estimated gradient $\bar{g}$ is accurate, we have a better bound as follows:
\begin{footnotesize}
\begin{align*}
	 & \mathbb{E}[f(\tilde{w}^{(T)})-f(w^*)]
     \le a_1^T\mathbb{E}[f(\tilde{w}^{(0)})-f(w^*)] + \frac{(a_1^T-1)a_2}{a_1-1}  \nonumber \\
     &\cdot (\frac{-8C^2-\frac{32C^2}{\min(\epsilon^l_i)}}{e^{\min(\epsilon^l_i)}-1} - \frac{8C^2}{(e^{\min(\epsilon^l_i)}-1)^2} + \frac{4C^2}{\min(\epsilon^l_i)^2})^{\frac{1}{2}} \nonumber \\
     &+ a_3(\frac{-8C^2-\frac{32C^2}{\min(\epsilon^l_i)}}{e^{\min(\epsilon^l_i)}-1} - \frac{8C^2}{(e^{\min(\epsilon^l_i)}-1)^2} + \frac{4C^2}{\min(\epsilon^l_i)^2}) \\
     &\le a_1^T\Delta 
    + \frac{a_1^T-1}{a_1-1}(O(a_2C/\min(\epsilon^l_i)) + O(a_3 C^2/\min(\epsilon^l_i)^2))
\end{align*}
\end{footnotesize}
The proof is completed.
\end{proof}

Given Lemma \ref{lemma:bound_after_cali}, the upper bound of convergence of the whole APES framework is obtained, as demonstrated in Theorem \ref{th_conv}.

\subsection*{C \quad Experiment Supplements}
In this section, we provide error evaluation for calibration during analyzing process of APES and S-APES (Fig. \ref{fig:calibration}), additional comparison on privacy amplification effects (Fig. \ref{fig:epsilon_005_05}, \ref{fig:epsilon_005_3}), and the output evaluation of Clip-Laplace Mechanism (Fig. \ref{fig:lap_var}, \ref{fig:lap_mean}). At last, the model accuracy and $\epsilon^{uc}$ for all the LDP settings is demonstrated in Tab. \ref{table:eps_acc_all}.
 \begin{table*}[!t]
   \centering
   \begin{footnotesize}
   \begin{tabular}{lllllllllllll}
     \toprule
    \multirow{2}*{{Framework}} & \multicolumn{2}{l}{\quad Uniform1} & \multicolumn{2}{l}{\quad Gauss1} & \multicolumn{2}{l}{\quad MixGauss1}  & \multicolumn{2}{l}{\quad Uniform2} & \multicolumn{2}{l}{\quad Gauss2} & \multicolumn{2}{l}{\quad MixGauss2}         \\
    \cline{2-13}
   \specialrule{0em}{1pt}{1pt}
      & $\epsilon^{uc}$ & acc & $\epsilon^{uc}$ & acc & $\epsilon^{uc}$ & acc & $\epsilon^{uc}$ & acc & $\epsilon^{uc}$ & acc & $\epsilon^{uc}$ & acc \\
     \midrule
     Non-Private & $\infty$ & \textbf{84.35\%} & $\infty$ & \textbf{84.35\%} & $\infty$ & \textbf{84.35\%} & $\infty$ & \textbf{84.35\%} & $\infty$ & \textbf{84.35\%} & $\infty$ & \textbf{84.35\%} \\
     LDP-min & 40.1 & 56.11\% & 40.1 & 56.11\% & 40.1 & 56.11\% & 40.1 & 56.11\% & 40.1 & 56.11\% & 40.1 & 56.11\% \\
     PLDP & 3925 & 75.64\% & 3925 & 63.77\% & 3925 & 67.75\% & 7850 & 77.54\% & 7850 & 64.81\% & 7850 & 73.84\% \\
     UniS & 17.4 & 75.64\% & 17.4 & 63.77\% & 17.5 & 67.75\% & 79.6 & 77.54\% & 76.9  & 64.81\% & 76.9  &  73.84\% \\
     APES & 15.5 & \textbf{78.77\%} & 15.3 & \textbf{78.33\%}  & 15.4 & \textbf{79.06\%} & 57.6 & \textbf{79.67\%} & 52.7 & \textbf{79.29\%} & 53.7  & \textbf{79.97\%} \\
     S-APES & \textbf{8.5} & 75.7\% & \textbf{8.5} & 56.64\% & \textbf{8.5} & 77.01\%  & \textbf{25.6}& 78.14\% & \textbf{23.6} & 73.4\% & \textbf{23.9} & 76.12\% \\ 
     \bottomrule
   \end{tabular}
   \end{footnotesize}
   \caption{Comparison with baseline frameworks on user-level central privacy budget and accuracy under the all the personalized privacy budget settings. $\delta^{uc}=3.6 \times 10^{-5}$. Total epochs: 40.}
   \label{table:eps_acc_all}
 \end{table*}

\begin{figure}[h] 
\centering
\includegraphics[scale=0.4]{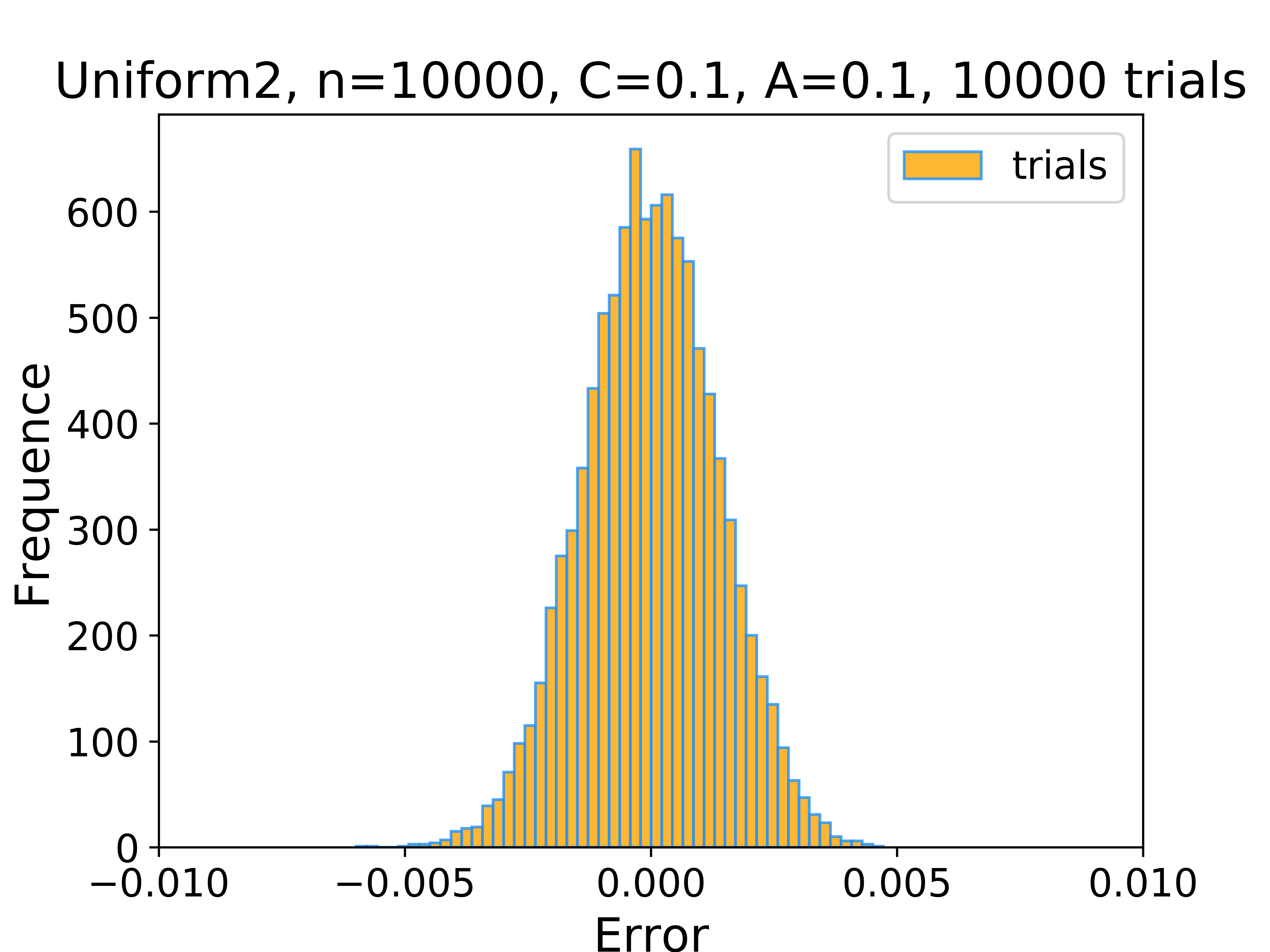} 
\caption{Error of approximation for calibrating noisy gradients. Most values of $(\tilde{g_i}-\mathbb{E}[\tilde{\bar{g_i}}])$ concentrate around zero.}
\label{fig:calibration}
\end{figure}

\begin{figure}[!t]
\centering
\includegraphics[scale=0.4]{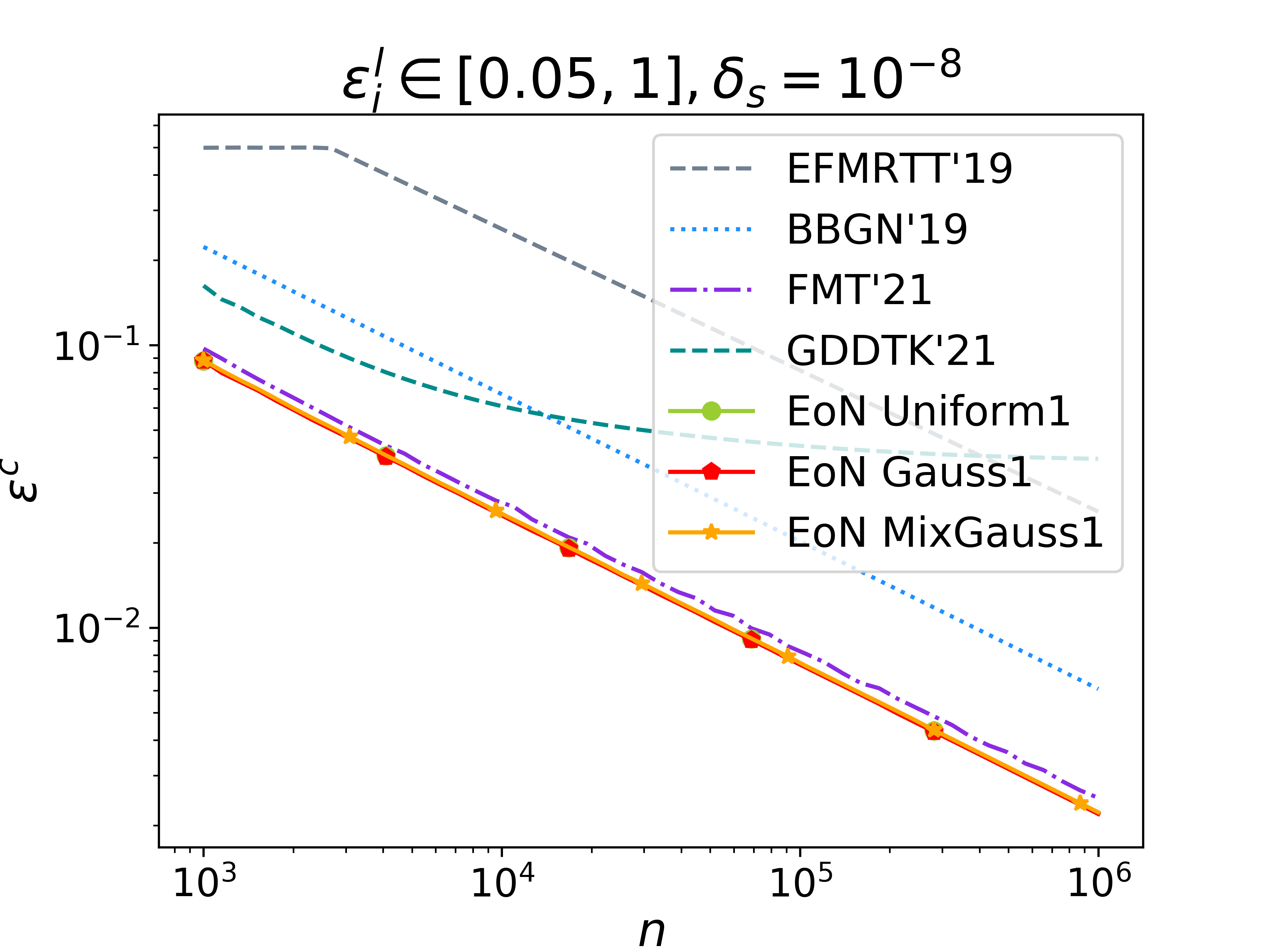} 
\caption{Privacy Bounds with smaller privacy budget range. The difference between EoN and baselines becomes smaller, but absolute value of $\epsilon^c$ of EoN is lower than larger range.}
\label{fig:epsilon_005_05}
\end{figure}

\begin{figure}[!t]
\centering
\includegraphics[scale=0.4]{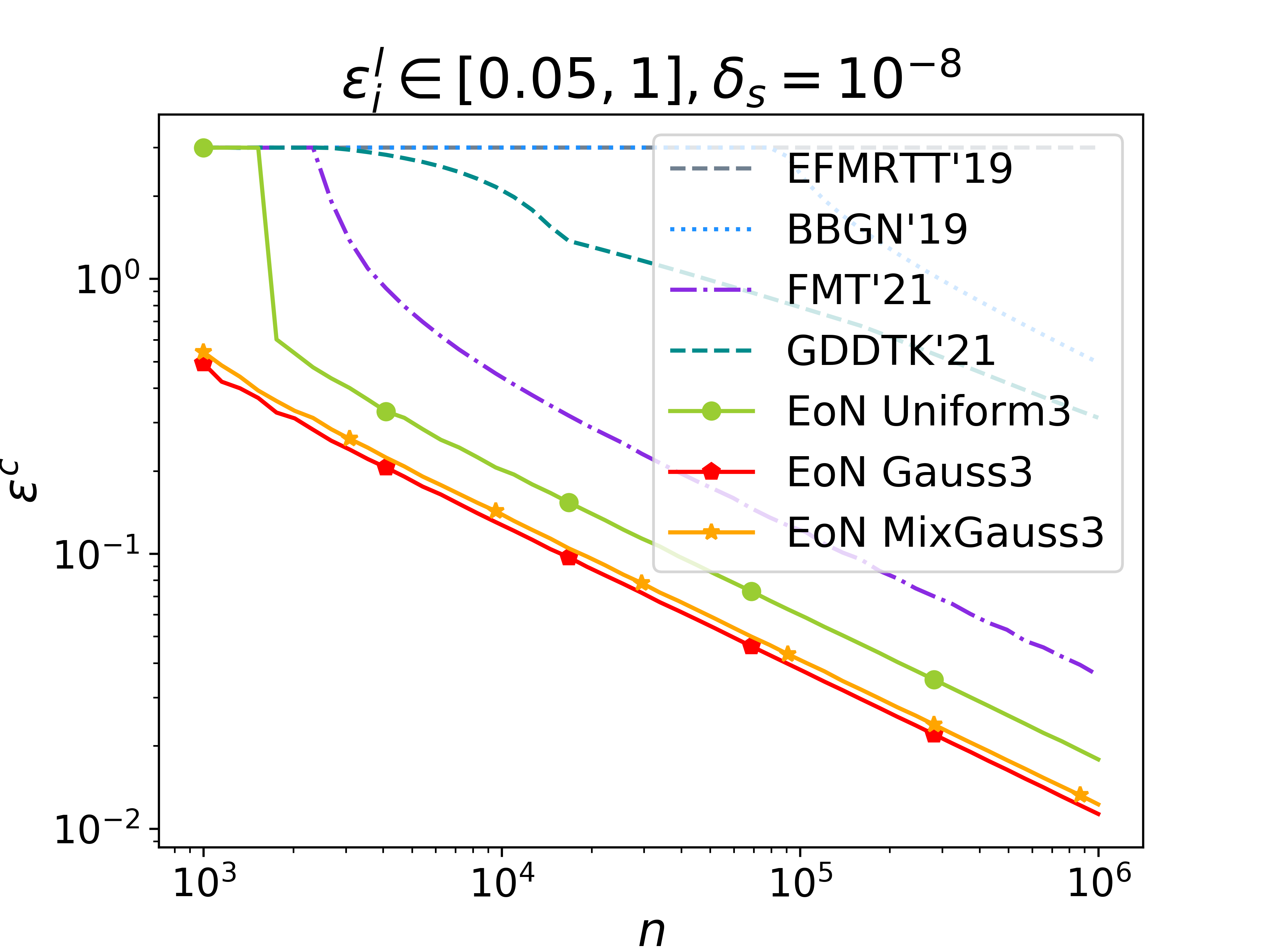} 
\caption{Privacy Bounds with larger privacy budget range. EoN gains more obvious amplification effect. Uniform3 denotes $\mathcal{U}(0.05,3)$, Gauss3 denotes $\mathcal{N}(0.5,1)$, MixGauss3 denotes 90\% $\mathcal{N}(0.5,1)$ and 10\% $\mathcal{N}(3,1)$. All the clip range of $\epsilon^l$ is in $[0.05, 3]$.}
\label{fig:epsilon_005_3}
\end{figure}

\begin{figure}[!t]
  \centering
  \includegraphics[scale=0.4]{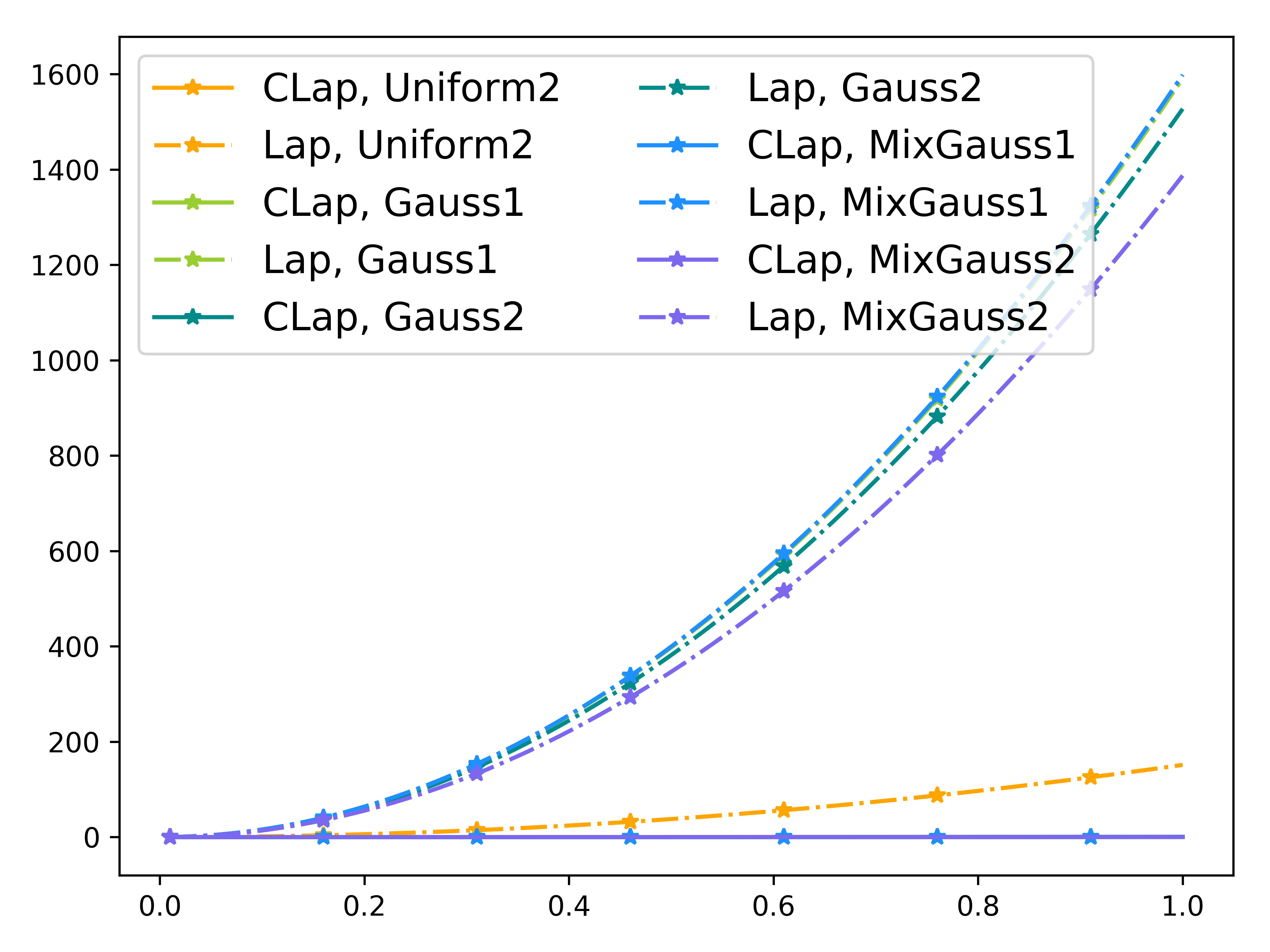}
  \caption{Variance of outputs of CLap on different $C$ and $\epsilon^l$. Smaller variance leads makes CLap more stable for varying parameters, which benefits from limited output ranges .}
  \label{fig:lap_var}
\end{figure}

\begin{figure}[!t]
\centering
\includegraphics[scale=0.4]{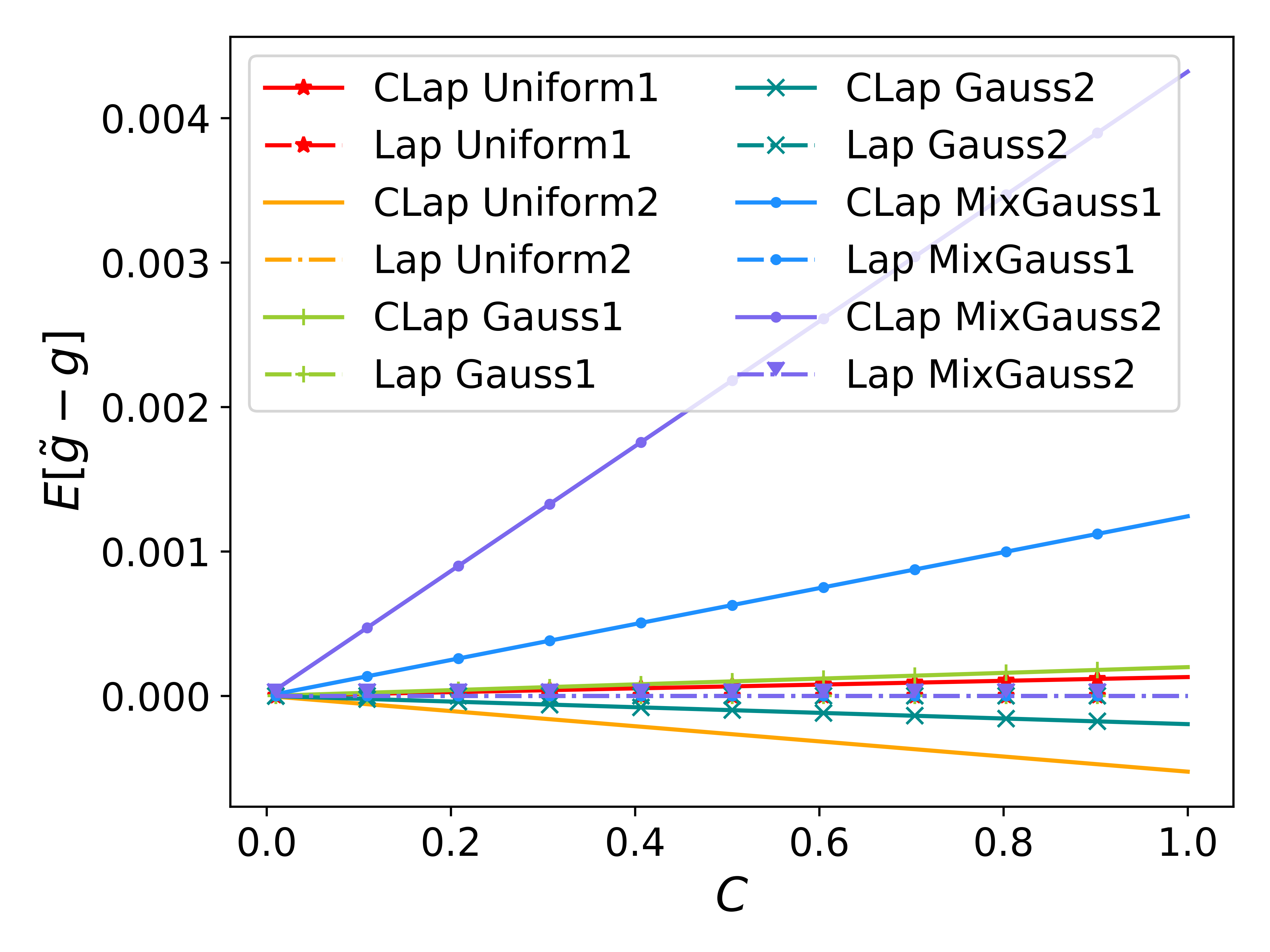} 
\caption{Bias of CLap outputs on different $C$ and $\epsilon^l$. Bias of CLap perturbation is larger than Lap perturbation, hence calibration is necessary.}
\label{fig:lap_mean}
\end{figure}

\end{document}